\documentclass[10pt, doublecolumn]{IEEEtran}
\pdfoutput=1
\usepackage{setspace}

\usepackage{cite}
\usepackage{indentfirst}
\usepackage{graphicx}
\usepackage{changepage}
\usepackage{stfloats}
\usepackage{amsfonts,amssymb}
\usepackage{bm}
\usepackage{algorithm}
\usepackage{algorithmic}
\usepackage{amsmath}
\usepackage{amsmath,amssymb,amsfonts}
\usepackage{times}
\usepackage{url}
\usepackage{cite}	
\usepackage{textcomp}
\usepackage{color}
\usepackage{setspace}
\usepackage{enumitem}
\usepackage{float}
\usepackage{diagbox}
\usepackage[T1]{fontenc}
\usepackage[utf8]{inputenc}
\usepackage{authblk}
\usepackage{threeparttable}
\usepackage{booktabs}
\usepackage{footnote}
\usepackage{graphicx}
\usepackage{pifont}
\usepackage[caption=false,font=footnotesize]{subfig}
\usepackage{mathrsfs}
\usepackage{bbm}
\usepackage{dsfont}
\usepackage{colortbl}
\usepackage{multirow}

\usepackage{makecell}
\usepackage[dvipsnames,table]{xcolor}

\allowdisplaybreaks[4]

\newenvironment{proof}{{\indent  \indent \it Proof:}}{\hfill $\blacksquare$}

\begin{document}
\title{Constellation Selection and Power Allocation for Multi-Cell OFDM-ISAC: Managing Inter-Cell Interference and Sensing Sidelobes}
\author{
	Kaitao Meng, \textit{Member, IEEE}, Kawon Han, \textit{Member, IEEE}, Christos Masouros, \textit{Fellow, IEEE}, and Lajos Hanzo, \textit{Life Fellow, IEEE}
	\thanks{Kaitao Meng is with the Department of Electrical and Electronic Engineering, University of Manchester, Manchester, UK (email: kaitao.meng@manchester.ac.uk). Kawon Han is with the Department of Electrical Engineering, Ulsan National Institute of Science and Technology (UNIST), Ulsan, South Korea (emails:kawon.han@unist.ac.kr). Christos Masouros is with the Department of Electronic and Electrical Engineering, University College London, London, UK (emails: c.masouros@ucl.ac.uk). L. Hanzo is with School of Electronics and Computer Science, University of Southampton, SO17 1BJ Southampton, UK (email: lh@ecs.soton.ac.uk) 
}}

\maketitle

\begin{abstract}
	Future integrated sensing and communication (ISAC) networks are expected to operate in dense multi-cell environments, where multiple base stations (BSs) share their time-frequency resources for communication and sensing. In such scenarios, the delay--Doppler (DD) sensing performance is strongly affected by random finite-alphabet orthogonal frequency-division multiplexing (OFDM) symbols, power allocation, receive filtering, and interference. This paper develops a modulation- and receive-filter-aware framework for the sensing-interference management in multi-cell OFDM-ISAC systems. Starting from a discrete-time OFDM sensing model, we derive closed-form signal-to-interference-plus-noise ratio (SINR) expressions for each range--Doppler bin under matched filtering (MF) and reciprocal filtering (RF). The analysis reveals distinct interference structures: MF depends on fourth-order constellation moments and power-overlap terms, whereas RF is governed by inverse-symbol-power and ratio-type interference terms. Based on these expressions, we obtain sensing-oriented power allocation structures, including a ramped water-filling solution for MF and a square-root allocation rule for RF. Furthermore, we jointly optimize the finite-alphabet constellation selection and power allocation under realistic communication and power constraints, and obtain tractable mixed-integer convex formulations for both MF and RF. Additionally, we study spectrum-overlap coordination in multi-cell scenarios and reveal the distinct MF/RF preferences for shared and orthogonalized tones. Furthermore, we extend the interference model to inter-cell propagation delays exceeding the cyclic prefix (CP), and show how the resultant delay violation redistributes the nominal interference spectrum into a delay-distorted effective spectrum. Numerical results validate the theoretical SINR expressions and demonstrate that joint constellation and power allocation is capable of reducing the DD-domain interference floor, enhancing weak-target visibility, and revealing the sensing--communication tradeoff in multi-cell OFDM-ISAC systems.
\end{abstract}


\newtheorem{thm}{\bf Lemma}
\newtheorem{remark}{\bf Remark}
\newtheorem{Pro}{\bf Proposition}
\newtheorem{theorem}{\bf Theorem}
\newtheorem{Assum}{\bf Assumption}
\newtheorem{Cor}{\bf Corollary}
\newtheorem{Def}{\bf Definition}
\newtheorem{assumption}{\bf Assumption}

\section{Introduction}

Next generation (NG) wireless networks are expected to go beyond `bit-pipe' connectivity and provide high-precision, robust, and ubiquitous sensing as well as localization capabilities \cite{10054381, Liu2022Integrated, Lu2024Integrated}. Integrated sensing and communication (ISAC) has therefore emerged as a key enabler, allowing the same radio infrastructure, spectrum, and waveforms to serve both data transmission and environmental perception \cite{Liu2020JointRadar,Zhang2021OverviewSignal,Liu2022SurveyFundamental}. In this context, communication-centric ISAC, where standard communication waveforms and frame structures are reused for sensing, is particularly attractive for near-term deployment because the existing cellular infrastructure already provides wide-area coverage and such designs exhibit compatibility with legacy cellular and WLAN-type systems \cite{Mishra2019Toward,shi2022device,Wei2023Integrated}. 

Among existing communication waveforms, orthogonal frequency-division multiplexing (OFDM) has become the default choice in cellular and WLAN standards. Leveraging OFDM for sensing, however, is far from trivial \cite{Zhang2021OverviewSignal,shi2022device}. Unlike radar probing relying on deterministic waveforms, a communication-centric OFDM sensing receiver must process random data-bearing symbols \cite{Lu2024Random, 11548574}. Under matched filtering (MF), this randomness appears in the ambiguity-function (AF) output as data-induced sidelobes, i.e., undesired off-peak responses away from the true target delay-Doppler bin. Strong sidelobes can raise the effective clutter floor, mask weak targets, or even generate ghost responses \cite{Chitre2020AFShaping}. Since practical constellations such as quadrature amplitude modulation (QAM) constellations generally exhibit a non-constant-modulus, the sidelobe and interference structures are strongly coupled with the modulation format \cite{Du2024Reshaping,Du2025PCSFiltering}. To suppress data-induced sidelobes, reciprocal filtering (RF) and more general mismatched filters (MMFs) have been proposed for separating or compensating for data symbols in the time-frequency (TF) domain, trading output signal-to-noise ratio (SNR) for improved sidelobe behavior \cite{Du2025PCSFiltering,han2025Constellation}. These MF/MMF choices interact with the modulation in a complex way, shaping the delay-Doppler profile, the interference floor, and ultimately the sensing performance.

The interaction between finite-alphabet modulation and OFDM sensing has only recently begun to receive explicit attention. For OFDM ISAC, the MF sensing performance can be characterized through constellation-dependent statistics, and probabilistic constellation shaping (PCS) can tune the trade-off between communication rate and sensing metrics such as the output SNR, integrated sidelobe ratio (ISLR), and sensing-channel mean-squared error (MSE) \cite{Du2024Reshaping,Du2025PCSFiltering}. A key insight is that ambiguity-function fluctuations are governed by higher-order symbol statistics: constant-modulus constellations such as phase-shift keying (PSK) usually induce lower fluctuations than high-order QAM, whose higher amplitude swings lead to stronger random MF sidelobes \cite{11206742}. Related studies further show that the effect of finite-alphabet signaling is receiver-dependent: MF-type processing is mainly governed by fourth-order moments, whereas RF-type processing is sensitive to inverse-symbol-power moments \cite{han2025Constellation,Meng2026Constellation}. Existing sidelobe-reduction or constellation-optimization methods may improve sensing performance through bespoke precoder design, sensing-oriented resource-block allocation, or alphabet reshaping \cite{Li2025MIMOOFDM,Li2025SensingOriented,Hu2025Learning,Geiger2025JCS,Du2023GCWkshpsPCS,Yang2024TWC_RandomWF}. However, many of these techniques require instantaneous transmit-signal optimization, dedicated sensing resources, or non-standard constellation redesign, which may require new signaling constellations, bit mappings, and link-adaptation tables. This motivates a standards-compliant alternative: retain the existing constellation set, select suitable constellations across subcarriers and scenarios, and jointly allocate power to satisfy the communication quality-of-service (QoS), while improving sensing.

Despite these advances, practical 6G deployments will rarely operate as isolated monostatic links. In dense cellular networks, multiple base stations (BSs) simultaneously reuse the TF resources, acting both as cooperative illuminators for target localization and as strong interferers to each other's sensing receivers \cite{Meng2024CooperativeTWC,han2025network,Meng2024CooperativeISACMag}. At the same time, both clutter and self-interference generate a structured interference-plus-sidelobe floor that is highly sensitive to the waveform, power allocation, and receiver filter. Existing ISAC localization and Cramér-Rao lower bound (CRLB) analyses usually model these effects through equivalent Gaussian noise and interference powers. Hence, they tend to be agnostic to the underlying modulation, subcarrier power profiles, and filtering structure, leading to suboptimal performance in practice \cite{Li2023Towards,Jing2024ISACSky,Yang2023Deployment,Hua20243DMultiTarget}. Conversely, waveform-centric studies on PCS and MF/MMF mainly consider a single link and quantify sensing performance through channel-state information MSE or dynamic range, without fully capturing multi-cell interference coupling and the corresponding target-estimation accuracy \cite{Du2024Reshaping}.

This creates a gap between filter-aware waveform design
and network-level ISAC sensing-interference management: we still lack a unified framework that (i) keeps track of the finite-alphabet OFDM signaling and of the TF-domain filter, (ii) explicitly accounts for clutter and multi-cell interference, and (iii) propagates these effects to delay--Doppler (DD)-domain sensing SINR and effective interference-plus-sidelobe (EISL) levels for multi-cell
target estimation. At the same time, the design space in a cellular ISAC network is significantly richer than in a single-link setting: each BS can select its constellation family, e.g., PSK-like versus high-order QAM, allocate power across subcarriers, and coordinate with neighboring BSs within a cluster \cite{Xie2023Networked,Meng2024CooperativeTWC,Meng2024CooperativeISACMag}. How these choices should be made jointly, under realistic interference and clutter, so as to balance communication performance and sensing SINR is still poorly understood. This issue is particularly important when the design is constrained to standardized constellation modes rather than relying on freely redesigned alphabets. In networked OFDM-ISAC, inter-cell interference may also arrive outside the cyclic prefix (CP) and cause delay-dependent spectral leakage, because such interference
is a one-way BS-to-receiver component and can remain strong even when
the desired monostatic echoes are CP-limited.

In this paper, we address the above knowledge gap by developing a modulation-aware framework for multi-cell OFDM ISAC in interference-limited networks. We commence from a discrete-time model of the OFDM sensing chain under MF and RF, including monostatic clutter and unknown multi-cell interference. Building on and extending the dynamic-range and MSE analyses in prior work \cite{Du2024Reshaping,han2025Constellation,Meng2026Constellation}, we derive closed-form expressions for the DD-bin sensing SINR and EISL as functions of second- or fourth-order constellation moments, power allocation, and filtering choice. These per-link metrics are then propagated to multi-cell delay and Doppler estimates. On top of this analytical foundation, we investigate how constellation selection and power allocation, both across subcarriers and across BSs within a cooperative cluster, can be optimized to exploit these new degrees of freedom.
The main contributions of this work are summarized as follows:

\begin{itemize}[leftmargin=1.2em]
	\item 
	We develop a multi-cell OFDM-ISAC framework accounting for both clutter and inter-cell interference. For MF and RF receivers, we derive closed-form sensing SINR expressions, explicitly revealing how the sensing performance depends on constellation geometry, subcarrier power allocation, AF sidelobes, and multi-cell interference. The analysis shows that MF is governed by fourth-order constellation moments and power-overlap terms, whereas RF is governed by inverse
	second-order constellation moments and ratio-type interference terms.
	\item 
	Based on the SINR expressions derived, we establish the closed-form  optimal power allocation laws for sensing-oriented design in multi-cell scenarios. Specifically, MF leads to a ramped water-filling solution that suppresses highly interfered tones, whereas RF yields a square-root allocation rule that assigns more power to tones with higher reciprocal-filter-induced interference and noise amplification. These results provide a theoretical explanation for the opposite interference-management behaviors of MF and RF.
	\item 
	We formulate the joint constellation selection and power allocation problem to optimize sensing performance under communication and reliability guarantees, and power constraints. Using an auxiliary-variable reformulation, the MF and RF design problems are cast as mixed-integer convex quadratic-over-linear programs. This connects finite-alphabet adaptive modulation with sensing SINR optimization and provides a tractable framework for studying the sensing--communication tradeoff.
	\item 
	We analyze the impact of inter-cell spectrum overlap and propagation-delay-induced leakage on sensing interference. We formally show when MF would favor full spectral overlap or frequency orthogonalization. Conversely, active-set-based RF favors reduced spectral overlap because its inter-cell interference penalty scales with the per-subcarrier interfering-to-sensing power ratio. We further derive the effective interference spectrum for beyond-CP interference and prove that one-symbol CP violation preserves total interference energy, but redistributes it across subcarriers.
\end{itemize}

\section{System Model}

\subsection{Transmit ISAC Signal Model}
We consider a multi-cell OFDM-ISAC system having $L$ co-channel BSs that share the same $N$ subcarriers, employing monostatic sensing at each BS. BS-0 is the reference BS whose sensing receiver is considered, while BSs indexed by $\ell$, $\ell=1,\ldots,L-1$ act as co-channel sensing and communication interferers. {{Unless otherwise stated, we consider local interference-aware adaptation at BS-0, where the transmit power profiles of the other BSs are treated as fixed over the design interval considered. BS-0 optimizes its own mode and power allocation in response to the resultant inter-cell interference. Cooperative inter-cell coordination is studied separately in Section~\ref{sec:multi_bs_power}.}}
 The unitary discrete Fourier transform (DFT) matrix is denoted by $\mathbf F_N\in\mathbb C^{N\times N}$, with $[\mathbf F_N]_{t,n}=N^{-1/2}e^{-j2\pi tn/N}$ for $t,n=0,\ldots,N-1$, and $\mathbf F_N^{H}$ denotes the corresponding inverse DFT (IDFT) matrix. Each BS transmits $M$ OFDM symbols, each having a symbol duration of $T_{\mathrm{sym}}$. 

For BS-$\ell$, the frequency-domain (FD) symbol on subcarrier $n$ and OFDM symbol index $m$ is written as
$
	X^{(\ell)}_m[n] \triangleq \sqrt{P^{(\ell)}_n}\, s^{(\ell)}_{n,m},
$
where $\{P^{(\ell)}_n\}$ are per-subcarrier powers, assumed constant across $m$ symbols, and the data symbols satisfy
$
	\mathbb E[s^{(\ell)}_{n,m}]=0$,
$	\mathbb E \big[|s^{(\ell)}_{n,m}|^2\big]=1.$
For BS-$\ell$, we collect the subcarrier symbols of the $m$-th OFDM symbol into
\begin{equation}
	\label{eq:Xvec_def}
	\vspace{-1mm}
	\mathbf X_m^{(\ell)} = \mathrm{diag}\big(\sqrt{P^{(\ell)}_0},\ldots,\sqrt{P^{(\ell)}_{N-1}}\big)\,\mathbf s_m^{(\ell)} .
	\vspace{-1mm}
\end{equation}
The per-BS power budget is
\begin{equation}
	\label{eq:avg_power_constraint}
	\vspace{-1mm}
	\frac{1}{N}\sum\nolimits_{n=0}^{N-1}P^{(\ell)}_n \;=\; P_{\mathrm{ave}}^{(\ell)},
	P^{(\ell)}_n \ge 0.
	\vspace{-1mm}
\end{equation}
To account for modulation-dependent sensing statistics, each subcarrier of BS-0 selects one of the modes from $j=0,1,\ldots,J$, where $j=0$ denotes a sensing-only mode without communication requirements, and $j\ge1$ denotes a communication constellation $\mathcal C_j$. The mode  $j=0$ allows a subcarrier that is not needed for data transmission to still contribute sensing energy. The selection variable $u_{n,j}\in \{0,1\}$ satisfies
\begin{equation}
	\vspace{-1mm}
	\sum\nolimits_{j=0}^{J}u_{n,j}=1, \forall n .
	\vspace{-1mm}
\end{equation}
For each mode, we define
\begin{equation}
	\label{moment_math}
	\vspace{-1mm}
	\mu_{4,j}\triangleq \mathbb E_{s\sim\mathcal C_j}\left[|s|^4\right],
	\qquad
	\mu_{-2,j}\triangleq \mathbb E_{s\sim\mathcal C_j}\left[|s|^{-2}\right].
	\vspace{-1mm}
\end{equation}
Then the effective moments on subcarrier $n$ of BS-0 are \cite{11202391}
\begin{equation}
	\vspace{-1mm}
	\mu_{4,n}^{(0)}= \! \sum\nolimits_{j=0}^{J} \!  u_{n,j}\mu_{4,j},
	\qquad
	\mu_{-2,n}^{(0)}= \! \sum\nolimits_{j=0}^{J} \!  u_{n,j}\mu_{-2,j}.
	\vspace{-1mm}
\end{equation}
This gives $\mu_{4,0}=\mu_{-2,0}=1$, corresponding to the moment-optimal unit-modulus sensing waveform. The transmitted symbols are assumed independent across
BSs, subcarriers, and OFDM symbols.

\begin{figure}[t]
	\centering
	\includegraphics[width=7.5cm]{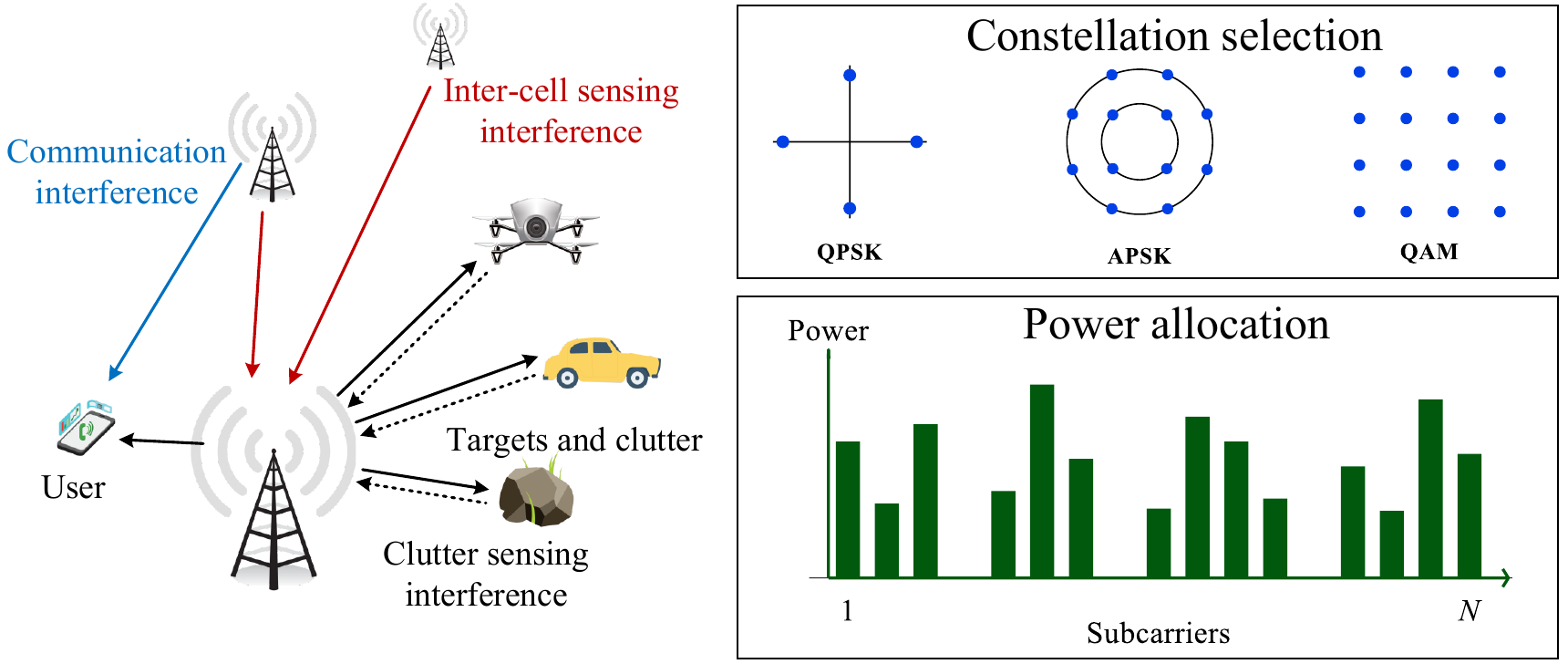}
	\vspace{-3mm}
	\caption{Multi-cell OFDM-ISAC interference management with constellation selection and power allocation.}
	\label{figure1}
\end{figure}
\subsection{Target and Interference Sensing Model}

For BS-$\ell$ to BS-0, the $q$-th propagation path is described by 
$\{\alpha_q^{(\ell,0)},\tau_q^{(\ell,0)},\nu_q^{(\ell,0)}\}$, where
$\alpha_q^{(\ell,0)}$, $\tau_q^{(\ell,0)}$, and $\nu_q^{(\ell,0)}$ denote the
complex path gain incorporating the target RCS (radar cross section), delay, and Doppler shift, respectively. Here, $\ell=0$ denotes the monostatic sensing links of BS-0 and $\ell\ge 1$ represents inter-cell sensing-interference links. The BS-0 monostatic path set includes both target echoes and clutter-scattering paths, which contribute through the same shifted DD response kernel in the subsequent per-bin SINR analysis. The index of $q=0$ is reserved for direct coupling or residual self-interference when present. For notational simplicity, we write
$\{\alpha_q^{(0)},\tau_q^{(0)},\nu_q^{(0)}\}
\equiv
\{\alpha_q^{(0,0)},\tau_q^{(0,0)},\nu_q^{(0,0)}\}$
for the BS-0 mono-static paths.

For BS-0, the dominant monostatic echoes are assumed to satisfy $\tau_q^{(0)}\le N_{\mathrm{CP}}$, so that the post-CP channel remains circular-shifted. Beyond-CP monostatic echoes are assumed negligible due to the substantial two-way path loss, whereas inter-cell interference may contain strong one-way components with $\tau_q^{(\ell,0)}>N_{\mathrm{CP}}$, as illustrated in Fig.~\ref{figure2}. For simplicity, we consider inter-symbol interference (ISI) caused by CP violation spanning at most one symbol. Higher delays can be treated similarly by shifting the contributing adjacent OFDM blocks.
Let $\mathbf y_m\in\mathbb C^{N}$ denote the CP-removed received block at BS-0, and let $\mathbf Y_m=\mathbf F_N\mathbf y_m$ with $Y_m[n]=[\mathbf Y_m]_n$. For a beyond-CP path, define
\begin{equation}
	\vspace{-1mm}
	d(\tau)\triangleq \tau-N_{\mathrm{CP}}.
	\vspace{-1mm}
\end{equation}
In the beyond-CP leakage analysis, \(d(\tau)\) is taken as an
integer sample delay. The CP-length shift contributes
only unit-modulus subcarrier phase factors in the FD.
Since the following analysis focuses on power-coupling coefficients, these
phase factors are omitted without changing
\(|[\mathbf G_i(d)]_{n,r}|^2\), \(i\in\{0,1\}\), defined in (\ref{eq:G0G1_def}). The matrices below
therefore explicitly model only the excess-delay-induced current-symbol
and previous-symbol mixing.
The FD ISI matrices are
\begin{equation}
	\label{eq:G0G1_def}
	\vspace{-1mm}
	\mathbf G_0(d)\triangleq \mathbf F_N \mathbf C_0(d)\mathbf F_N^{H},\qquad
	\mathbf G_1(d)\triangleq \mathbf F_N \mathbf C_1(d)\mathbf F_N^{H},
	\vspace{-1mm}
\end{equation}
where 
$
	\mathbf C_0(d)=
	\begin{bmatrix}
		\mathbf 0_{d\times (N-d)} & \mathbf 0_{d\times d}\\
		\mathbf I_{N-d} & \mathbf 0_{(N-d)\times d}
	\end{bmatrix},
$
$
	\mathbf C_1(d)=
	\begin{bmatrix}
		\mathbf 0_{d\times (N-d)} & \mathbf I_d\\
		\mathbf 0_{(N-d)\times (N-d)} & \mathbf 0_{(N-d)\times d}
	\end{bmatrix}.
$
Here, $\mathbf C_0(d)$ and $\mathbf C_1(d)$ represent the current-symbol and previous-symbol ISI contributions, respectively.
The aggregate TF operator encompassing beyond-CP delays is formulated as
\begin{equation}
	\label{eq:U_freq_def}
	\begin{aligned}
		\vspace{-1mm}
		&\mathcal U^{\mathrm f}_{\tau,\nu}(\mathbf X_m,\mathbf X_{m-1})
		\triangleq\; \\
		&\begin{cases}
			\mathbf D(\tau)\mathbf X_m,
			& 0\le \tau \le N_{\mathrm{CP}},\\[1mm]
			\begin{aligned}[t]
				&\mathbf G_0(d(\tau))\mathbf X_m +  \\
				& e^{-j2\pi \nu T_{\mathrm{sym}}}
				\mathbf G_1(d(\tau))\mathbf X_{m-1},
			\end{aligned}
			& N_{\mathrm{CP}} < \! \tau \!  \le N_{\mathrm{CP}}+N-1 .
		\end{cases}
		\vspace{-1mm}
	\end{aligned}
\end{equation}
where $
\mathbf D(\tau)\triangleq \mathrm{diag}\!\left(1,e^{-j2\pi\frac{1}{N}\tau},\ldots,e^{-j2\pi\frac{N-1}{N}\tau}\right).
$
Then, we have
	\begin{align}
		\label{eq:Y_model_RD}
		\vspace{-1mm}
		&Y_m[n] 
		=
		\underbrace{
			\sum_{q=0}^{Q_0} \alpha_q^{(0)} e^{j2\pi\nu_q^{(0)}m T_{\mathrm{sym}}}
			\big[\mathbf D(\tau_q^{(0)})\mathbf X_m^{(0)}\big]_n
		}_{\text{BS-0 monostatic returns}}
		\\
		&+
		\underbrace{
			\sum_{\ell=1}^{L-1}
			\sum_{q=0}^{Q_\ell}\alpha_q^{(\ell,0)} e^{j2\pi\nu_q^{(\ell,0)} m T_{\mathrm{sym}}}
			\big[\mathcal U^{\mathrm f}_{\tau_{q}^{(\ell,0)},\nu_q^{(\ell,0)}}
			(\mathbf X_m^{(\ell)},\mathbf X_{m-1}^{(\ell)})\big]_n
		}_{\text{interference from other BSs}} \nonumber
		\\
		&+
		{Z_m[n]} ,  \nonumber
		\vspace{-1mm}
	\end{align}
where $Z_m[n]=[\mathbf F_N\mathbf z_m]_n\sim\mathcal{CN}(0,\sigma_z^2)$, which is independently and identically distributed (i.i.d.) across $(n,m)$.

\subsection{DD-Domain Sensing Receive-Filter Output}
\label{DDdomain}
\vspace{-1mm}
Range--Doppler processing applies an element-wise weight $V_m[n]$ to each TF sample. We consider MF and RF weights
\begin{equation}\label{eq:V_MF_RF}
	\vspace{-1mm}
	V_m[n]=
	\begin{cases}
		X^{(0)*}_m[n], & \text{MF},\\[1mm]
		1/X^{(0)}_m[n], & \text{RF},
	\end{cases}
\vspace{-1mm}
\end{equation}
on active tones, and $V_m[n]=0$ otherwise.
The DD response is
\begin{equation}\label{eq:Lambda_DD_def}
	\vspace{-1mm}
	\Lambda[k,p]\triangleq \frac{1}{\sqrt{NM}}
	\sum_{m=0}^{M-1}\sum_{n=0}^{N-1}
	Y_m[n]V_m[n]\,
	e^{j\frac{2\pi}{N}nk}e^{-j\frac{2\pi}{M}mp}.
	\vspace{-1mm}
\end{equation}
Substituting \eqref{eq:Y_model_RD} into \eqref{eq:Lambda_DD_def} yields
\begin{equation}\label{eq:Lambda_DD_decomp}
	\vspace{-1mm}
	\Lambda[k,p]
	=
	\underbrace{\Lambda_{\rm Sen}[k,p]}_{\text{Mono-static sensing}}
	+
	\underbrace{\Lambda_{\rm Inter}[k,p]}_{\text{Interference}}
	+
	\underbrace{\Lambda_z[k,p]}_{\text{noise}},
	\vspace{-1mm}
\end{equation}
where
$
	\Lambda_{\rm Sen}[k,p]
	\triangleq
	\sum_{q=0}^{Q_0}\alpha_q^{(0)}
	\chi_q^{(0,0)}[k,p],
$
$
	\Lambda_{\rm Inter}[k,p]
	\triangleq
	\sum_{\ell=1}^{L-1}\sum_{q=0}^{Q_\ell}\alpha_q^{(\ell,0)}
	\chi_q^{(\ell,0)}[k,p],
$
and
$
	\Lambda_z[k,p]
	\triangleq
	\frac{1}{\sqrt{NM}}
	\sum_{m=0}^{M-1}\sum_{n=0}^{N-1}
	Z_m[n]V_m[n]\,
	e^{j\frac{2\pi}{N}nk}e^{-j\frac{2\pi}{M}mp}.
$
Here $\chi_q^{(\ell,0)}[k,p]$ denotes the weighted DD response of the $q$-th path arriving from BS-$\ell$:
\begin{equation}\label{eq:chi_def}
	\chi_q^{(\ell,0)}[k,p]
	\triangleq
	\frac{1}{\sqrt{NM}}
	\sum_{m=0}^{M-1}\sum_{n=0}^{N-1}
	\tilde X_{q,m}^{(\ell)}[n]\,V_m[n]\,
	e^{j\frac{2\pi}{N}nk}e^{-j\frac{2\pi}{M}mp},
\end{equation}
where $\tilde X_{q,m}^{(\ell)}[n]$ is the effective TF contribution of the $q$-th path.

\begin{figure}[t]
	\centering
	\includegraphics[width=7.5cm]{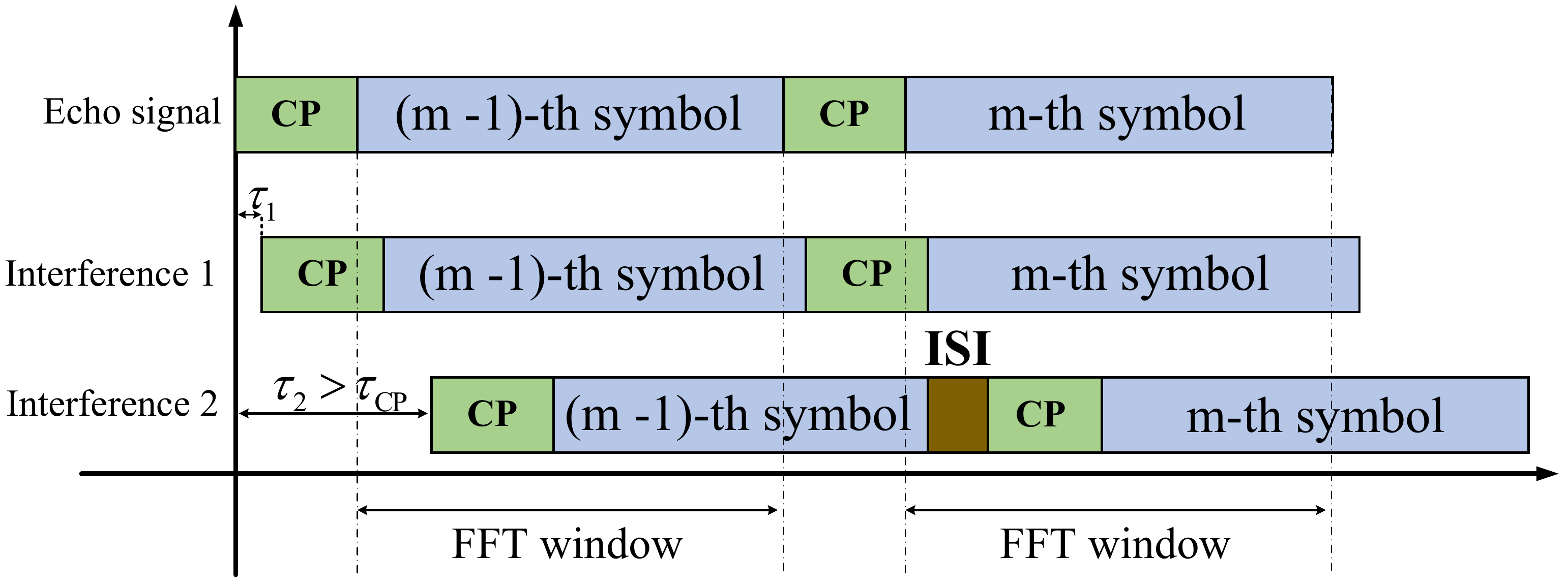}
	\vspace{-3mm}
	\caption{Timing structure of within-CP and beyond-CP inter-cell interference.}
	\label{figure2}
\end{figure}

\subsection{Communication QoS Model}

We impose a per-subcarrier communication QoS constraint for a representative downlink user associated with BS-0. The effective communication SINR on subcarrier $n$ is modeled as
\begin{equation}
	\label{eq:DL_SINR_simple}
	\vspace{-1mm}
	\gamma_n^{\mathrm c}
	=
	\frac{P_n^{(0)}|h_n^{\mathrm c}|^2}
	{\sigma_w^2+\sum_{\ell=1}^{L-1}\beta_n^{(\ell,\mathrm c)}P_n^{(\ell)}} ,
	\vspace{-1mm}
\end{equation}
where $h_n^{\mathrm c}$ is the desired channel, $\sigma_w^2$ is the noise power at the communication receiver, and $\beta_n^{(\ell,\mathrm c)}$ captures the effective inter-cell interference coupling, including possible timing misalignment and beyond-CP leakage. Define
$
	g_n
	\triangleq
	\frac{|h_n^{\mathrm c}|^2}
	{\sigma_w^2+\sum_{\ell=1}^{L-1}\beta_n^{(\ell,\mathrm c)}P_n^{(\ell)}} ,
$
so that
$
	\gamma_n^{\mathrm c}=g_nP_n^{(0)}.
$

{{Let $r_j=\log_2|\mathcal C_j|$ denote the number of mapped bits per
		modulation symbol of mode $j$. Since a modulated symbol transmitted
		on an OFDM subcarrier corresponds to a complex channel use, $r_j$
		is equivalently measured in bits per channel use under the
		OFDM resource-element model adopted. Following \cite{Meng2026Constellation}, each communication mode $j$ is
		associated with a constellation-dependent SINR threshold
		$
			\Gamma_j=\Gamma(\mathcal C_j,\epsilon_{\rm tar})
		$
		in order to meet a specific target BER $\epsilon_{\rm tar}$. The thresholds can be obtained from the corresponding BER-versus-SINR curves of the selected constellations. Naturally, higher-order constellations generally require higher SINR to achieve the same target BER.
		In exchange for a higher throughput, we set $r_0=\Gamma_0=0$ for the sensing-only mode. Accordingly, we define the average payload per channel use supported as}}
\begin{equation}
	\vspace{-1mm}
	R_{\mathrm{com}}^{(0)}
	=
	\frac{1}{N}\sum\nolimits_{n=0}^{N-1}\sum\nolimits_{j=1}^{J}u_{n,j}r_j
	\ge R_{\min}.
	\vspace{-1mm}
\end{equation}
The reliability constraint on communication-active subcarriers is
$
	\gamma_n^{\mathrm c}
	\ge
	\sum_{j=1}^{J}u_{n,j}\Gamma_j,
$
or equivalently,
$
	P_n^{(0)}
	\ge
	\sum_{j=1}^{J}u_{n,j}\frac{\Gamma_j}{g_n}.
$
Thus, if mode $j$ is selected on subcarrier $n$, the minimum power required is
$
	P_{n,j}^{\min}=\frac{\Gamma_j}{g_n}, j=1,\ldots,J.
$
{{For the sensing-only mode, no communication-driven minimum power is
		required.}}

\subsection{Problem Formulation}

{{For the given inter-cell interference profiles, we jointly optimize the mode selection and power allocation of BS-0 over the $N$ subcarriers for maximizing its sensing metric, while satisfying its communication QoS constraints.}} A closed-form sensing SINR is not yet available because the DD-domain interference depends jointly on the modulation statistics, power allocation, and receive filter. The problem is formulated as:
\begin{subequations}\label{prob:P1}
	\begin{align}
		(\mathrm{P1})\quad
		\max_{{\bm{U}},{\mathbf P}^{(0)}} \quad 
		& \Gamma_s({\bm{U}},{\mathbf P}^{(0)})  \label{P1:obj}\\
		\text{s.t.} 
		& 
		\frac{1}{N}
		\sum_{n=0}^{N-1}\sum_{j=1}^{J}u_{n,j}r_j
		\ge R_{\min}, \label{P1:rate} \\
		&   
		\gamma^{\rm c}_n
		\ge
		\sum_{j=1}^{J} u_{n,j}{\Gamma}_j,
		\forall n, \label{P1:ber}\\
		& 
		\frac{1}{N}\sum_{n=0}^{N-1} P_n^{(0)}  = P_{\mathrm{ave}}^{(0)},
		 P_n^{(0)} \in [0,P_{\max}], \forall n, \label{P1:power}\\
		& 
		\sum_{j=0}^{J} u_{n,j}=1,
		u_{n,j}\in\{0,1\},
		\forall n,\ j. \label{P1:select}
	\end{align}
\end{subequations}
Here, $q$ indexes the target path of interest, ${\bm U}=\{u_{n,j}\}$, $\mathbf P^{(0)}=\{P_n^{(0)}\}$, and $\Gamma_s({\bm U},\mathbf P^{(0)})\in \{ \mathrm{SINR}_q^{\mathrm{MF}},\mathrm{SINR}_q^{\mathrm{RF}}\}$ denotes the sensing output SINR after MF or RF processing, whose closed-form expressions are derived in Section III. The resultant design may then be
evaluated for all target range-Doppler bins. Constraints \eqref{P1:rate} and \eqref{P1:ber} impose the communication payload and reliability requirements, respectively; \eqref{P1:power} imposes the average-power and per-subcarrier peak-power constraints; and \eqref{P1:select} ensures that each subcarrier selects either the sensing-only mode of $j=0$ or a communication constellation from $j=1,\ldots,J$. The resultant problem is a mixed-integer nonlinear program (MINLP), since the constellation-dependent moments $\mu_{4,n}^{(0)}$ and $\mu_{-2,n}^{(0)}$, together with the power allocation, jointly affect the sensing metric and communication feasibility.

\section{Sensing Performance Analysis Within CP}
\vspace{-1.5mm}
{Existing studies mainly characterize modulation-induced sensing effects through AF sidelobe, dynamic-range, or estimation-MSE metrics, primarily for individual sensing links. Here, we derive DD-bin SINR expressions in multi-cell interference environments, revealing the power-overlap coupling of MF and the ratio-type interference coupling of RF that directly motivate the subsequent resource-allocation analysis.}

We first analyze a baseline case in which all paths lie within the CP. Specifically, the paths---including inter-cell interference
considered in this section---satisfy
$0\le \tau_q^{(\ell,0)}\le N_{\mathrm{CP}}$ with integer sample-spaced
delays, and their Doppler shifts are aligned with the coherent processing
grid, i.e.,
$
\nu_q^{(\ell,0)}T_{\mathrm{sym}}
=
\frac{\psi_q^{(\ell,0)}}{M},
\psi_q^{(\ell,0)}\in\{0,\ldots,M-1\}.
$
Under this model, each path induces a circular delay--Doppler shift of the
corresponding weighted kernel. Thus, we have
\begin{equation}
	\vspace{-1mm}
	\chi_q^{(\ell,0)}[k,p]
	=
	\chi^{(\ell,0)}
	\!\big[k-\tau_q^{(\ell,0)},\,p-\psi_q^{(\ell,0)}\big],
	\vspace{-1mm}
\end{equation}
where
\begin{equation}
	\vspace{-1mm}
	\chi^{(\ell,0)}[k,p]
	\triangleq
	\frac{1}{\sqrt{NM}}
	\sum_{m=0}^{M-1}\sum_{n=0}^{N-1}
	X_m^{(\ell)}[n]V_m[n]
	e^{j\frac{2\pi}{N}nk}e^{-j\frac{2\pi}{M}mp}.
	\vspace{-1mm}
\end{equation}

The path coefficients are independent of the data symbols and noise.
Throughout the SINR derivation, the path powers
$|\alpha_q^{(\ell,0)}|^2$ are treated as fixed large-scale quantities,
whereas the path phases are modeled as independent random phases. Equivalently,
different propagation paths are uncorrelated in the ensemble sense, i.e.,
$\mathbb E[\alpha_q^{(\ell,0)}\alpha_{q'}^{(\ell',0)*}]=0$ for
$(\ell,q)\neq(\ell',q')$. This removes cross-path terms in the averaged
DD-bin power, while retaining the conditioned path powers
$|\alpha_q^{(\ell,0)}|^2$ in the SINR expressions. The symbols are independent
both across BSs and across TF indices $(n,m)$.
Accordingly,
$
{\mathbb E}[\Lambda_{\rm Sen}\Lambda_{\rm Inter}^*]=0$,
$\mathbb E[\Lambda_{\rm Sen}\Lambda_{z}^*]=0,
\mathbb E[\Lambda_{\rm Inter}\Lambda_{z}^*]=0,
$
so that the total DD-bin power decomposes as
\begin{equation}\label{eq:Lambda_power_split}
	\vspace{-1mm}
	\mathbb E|\Lambda[k,p]|^2=
	\mathbb E|\Lambda_{\rm Sen}[k,p]|^2+
	\mathbb E|\Lambda_{\rm Inter}[k,p]|^2+
	\mathbb E|\Lambda_z[k,p]|^2.
	\vspace{-1mm}
\end{equation}

Let $\zeta^{(\ell)}_{n,m}\triangleq \tilde X^{(\ell)}_m[n]V_m[n]$. Then, under the adopted symbol model, the weighted samples $\{\zeta_{n,m}^{(\ell)}\}$ are assumed mutually uncorrelated across distinct TF indices.
Under the uncorrelated-grid model, we have
	\begin{align}\label{eq:chi_second_compact}
		\vspace{-1mm}
		\mathbb E\!\left[|\chi^{(\ell,0)}[k,p]|^2\right]
		&=\frac{1}{NM}\sum_{n,m}\mathrm{Var}(\zeta^{(\ell)}_{n,m}) \\
		&+\frac{1}{NM}\left|\sum_{n,m}\mathbb E[\zeta^{(\ell)}_{n,m}]
		e^{j\frac{2\pi}{N}nk}e^{-j\frac{2\pi}{M}mp}\right|^2. \nonumber
		\vspace{-1mm}
	\end{align}
For $\ell\ge1$, we have $\mathbb E[\tilde X^{(\ell)}_m[n]]=0$ and independence implies
$\mathbb E[\zeta^{(\ell)}_{n,m}]=0$, so \eqref{eq:chi_second_compact} reduces to a
flat (approximately $[k,p]$-independent) kernel power.
The filtered noise term satisfies
\begin{equation}\label{eq:P_noise}
	\vspace{-1mm}
	\mathbb E|\Lambda_z[k,p]|^2
	=\frac{\sigma_z^2}{NM}\sum_{n,m}\mathbb E|V_m[n]|^2.
	\vspace{-1mm}
\end{equation}

\begin{figure*}[t]
	\footnotesize
	\begin{equation}
		\label{SINR_derived_expression}
		\mathrm{SINR}^{\mathrm{MF}}_q
		\! = \!
		\frac{
			|\alpha_q^{(0)}|^2
			\mathbb E \big[|\chi^{(0,0)}_{\mathrm{MF}}[0,0]|^2\big]
		}
		{
			\underbrace{
				\sum_{\substack{q'=0,q'\neq q}}^{Q_0}
				|\alpha_{q'}^{(0)}|^2
				\frac{
					M\sum_{n=0}^{N-1}\big(P_n^{(0)}\big)^2\mu^{(0)}_{4,n}
					\! - \! {MN} \big(P_{\mathrm{ave}}^{(0)}\big)^2
					\! - \! \frac{1}{N}\sum_{n=0}^{N-1}
					\big(P_n^{(0)}\big)^2\big(\mu^{(0)}_{4,n}-1\big)
				}{NM-1}
			}_{\substack{\text{sidelobes due to other targets and clutter}}}
			\! + \!
			\underbrace{
				\frac{1}{N}\sum_{\ell=1}^{L-1}
				\sum_{q'=0}^{Q_\ell}|\alpha_{q'}^{(\ell,0)}|^2
				\sum_{n=0}^{N-1} P_n^{(0)}P_n^{(\ell)}
			}_{\text{inter-cell interference}}
			\! + \!
			\underbrace{
				\sigma_z^2 P_{\mathrm{ave}}^{(0)}
			}_{\text{noise}}
		}.
	\end{equation}
\end{figure*}

\subsection{SINR of Matched Filtering}
\label{sec:mf_sinr}

Using the DD-domain model in
\eqref{eq:Lambda_DD_def}--\eqref{eq:Lambda_DD_decomp}, consider the $q$-th
monostatic path of BS-0 with DD pair
$(\tau_q^{(0)},\psi_q^{(0)})$. On the matched DD bin
$[k,p]=(\tau_q^{(0)},\psi_q^{(0)})$, the MF receive weight is
$V_m[n]=X^{(0)*}_m[n]$, and the DD-domain output can be decomposed
as
\begin{equation}
	\label{eq:Lambda_MF_decomp_q}
	\begin{aligned}
		\vspace{-1mm}
		&\Lambda_{\mathrm{MF}}[\tau_q^{(0)},\psi_q^{(0)}]
		= {\alpha_q^{(0)} \chi^{(0,0)}_{\mathrm{MF}}[0,0]}  \\[2pt]
		&\quad + {\sum_{\substack{q'=0\\q'\ne q}}^{Q_0}
			\alpha_{q'}^{(0)} \chi^{(0,0)}_{\mathrm{MF}} \big[\tau_q^{(0)}-\tau_{q'}^{(0)},\,
			\psi_q^{(0)}-\psi_{q'}^{(0)}\big]} \\[2pt]
		&\quad + {\sum_{\ell=1}^{L-1}
			\sum_{q'=0}^{Q_\ell}\alpha_{q'}^{(\ell, 0)}
			\chi^{(\ell,0)}_{\mathrm{MF}} \big[\tau_q^{(0)}-\tau_{q'}^{(\ell, 0)},\,
			\psi_q^{(0)}-\psi_{q'}^{(\ell, 0)}\big]}\\
		&\quad + \Lambda_{z,\mathrm{MF}}[\tau_q^{(0)},\psi_q^{(0)}],
		\vspace{-1mm}
	\end{aligned}
\end{equation}
where $\chi^{(\ell,0)}_{\mathrm{MF}}[k,p]$ is the MF cross-ambiguity kernel obtained with $V_m[n]=X^{(0)*}_m[n]$.

The DD-bin SINR  after filtering can be defined as
\begin{equation}\label{eq:SINR_MF_bin_def}
	\vspace{-1mm}
	\mathrm{SINR}^{\mathrm{MF}}_q
	\triangleq
	\frac{\mathbb E\!\left[|\alpha_q^{(0)}|^2\big|\chi^{(0,0)}_{\mathrm{MF}}[0,0]\big|^2\right]}
	{\mathbb E\!\left[\left|\Lambda_{\mathrm{MF}}[\tau_q^{(0)},\psi_q^{(0)}]
		-\alpha_q^{(0)}\chi^{(0,0)}_{\mathrm{MF}}[0,0]\right|^2\right]}.
		\vspace{-1mm}
\end{equation}
Since the transmitted waveform is known at BS-0,
$\alpha_q^{(0)}\chi_{\rm MF}^{(0,0)}[0,0]$ is treated as the useful
data-aided coherent component on the matched DD bin. The denominator in
\eqref{eq:SINR_MF_bin_def} therefore collects the residual undesired monostatic paths, inter-cell
interference paths, and filtered noise terms.

\begin{Pro}\label{MFSINR_Expression}
	Under the uncorrelated-grid model and the uniform-offset approximation for
	undesired monostatic paths, the MF per-bin SINR in the DD domain is
	approximated by (\ref{SINR_derived_expression}), as shown at the top of this page.
\end{Pro}

\begin{proof}
	Please  refer to Appendix A.
\end{proof}

This expression coincides with the 1-D range-domain result and it is now explicitly interpreted as the SINR on the
DD bin $(\tau_q^{(0)},\psi_q^{(0)})$ obtained from
$\Lambda[k,p]$.

\subsection{SINR of Reciprocal Filtering}
\label{sec:rf_sinr}

In the DD-domain framework, RF uses $V_m[n]=1/X_m^{(0)}[n]$ on sensing-active tones. Unless otherwise stated, we adopt the full-band RF model, where all $N$ sensing subcarriers of BS-0 remain active and $P_n^{(0)}>0$. Thus, communication inactivity does not imply sensing silence. 

Let $\chi_{\mathrm{RF}}^{(\ell,0)}[k,p]$ denote the RF
cross-ambiguity kernel obtained using $V_m[n]=1/X_m^{(0)}[n]$.
Substituting this weight into \eqref{eq:Lambda_MF_decomp_q} gives the
same desired-path, undesired-monostatic, inter-cell-interference, and
noise decomposition, with $\chi_{\mathrm{MF}}^{(\ell,0)}$ replaced by
$\chi_{\mathrm{RF}}^{(\ell,0)}$.

\begin{Pro}\label{SensingSINRRF}
	The RF per-bin SINR is given by
	\begin{equation}
		\label{eq:SINR_RF_final_DD}
		\begin{aligned}
			\vspace{-1mm}
			\mathrm{SINR}^{\mathrm{RF}}_q
			\!	= \!
			\frac{|\alpha_q^{(0)}|^2\, M N^2}
			{\displaystyle
				\sigma_z^2\sum_{n=0}^{N-1} \! \frac{\mu^{(0)}_{-2,n}}{P^{(0)}_n}
				\! + \!
				\sum_{\ell=1}^{L-1}\Big(\sum_{q'}|\alpha_{q'}^{(\ell, 0)}|^2\Big)
				\! \sum_{n=0}^{N-1} \! \frac{P^{(\ell)}_n}{P^{(0)}_n}\mu^{(0)}_{-2,n}}.
				\vspace{-1mm}
		\end{aligned}
	\end{equation}
\end{Pro}

\begin{proof}
	Please refer to Appendix B.
\end{proof}

\begin{remark}
	{{Intuitively, MF and RF exhibit opposing interference
			trends. Explicitly, under MF, allocating more power to an interfered tone increases
			the power-overlap term and hence the interference. By contrast, under RF,
			a lower $P_n^{(0)}$ increases the reciprocal-filter gain
			$1/X_m^{(0)}[n]$ and augments the interference and noise on that tone,
			as reflected by the ratio $P_n^{(\ell)}/P_n^{(0)}$.}}  This contrast will be made explicit through the sensing-only special cases below.
\end{remark}

\subsection{Special Cases: Optimal Sensing without Communication Constraints}
In this section, we focus on sensing-oriented interference management and temporarily disregard the communication-related constraints, so as to isolate the fundamental sensing trade-off induced by inter-cell spectrum coupling.

\subsubsection{Optimal Power Allocation for MF}

Under the equality sum-power constraint $\sum_{n=0}^{N-1}P_n^{(0)}=N P_{\mathrm{ave}}^{(0)}$, the MF noise term is independent of the power profile. The MF denominator-minimization surrogate is therefore
\begin{equation}
	\label{eq:MF_denominator_reduced}
	\begin{aligned}
		\vspace{-1mm}
		\widetilde{\mathcal D}^{\mathrm{MF}}\big(\mathbf P^{(0)};\{\mathbf P^{(\ell)}\}\big)
		&=
		S^{(0)}_q\sum_{n=0}^{N-1} b_n^{(0)} \big(P_n^{(0)}\big)^2
		\\ 
		&+
		\frac{1}{N}\sum_{\ell\neq 0}\Xi_{0,\ell}\sum_{n=0}^{N-1}P_n^{(0)}P_n^{(\ell)}.
		\vspace{-1mm}
	\end{aligned}
\end{equation}
where
$b_n^{(0)}
\triangleq
\frac{
	M\mu_{4,n}^{(0)}-\frac{\mu_{4,n}^{(0)}-1}{N}
}{NM-1}$, $S^{(0)}_q \triangleq \sum_{q' \neq q} \left|\alpha_{q'}^{(0)}\right|^2,
$ and 
$	\Xi_{0,\ell} \triangleq \sum_{q''=0}^{Q_\ell}
	|\alpha_{q''}^{(\ell, 0)}|^2.
$
{{Given $\{P_n^{(\ell)}\}_{\ell\neq0}$, we minimize this convex
		denominator surrogate subject to \eqref{eq:avg_power_constraint}.
		Under the fixed sum-power constraint, the dominant coherent component
		of the MF numerator is $MN(P_{\mathrm{ave}}^{(0)})^2$ and it is independent
		of the power profile. Its remaining profile-dependent term is bounded by
		$(\mu_{4,\max}-1)P_{\max}P_{\mathrm{ave}}^{(0)}$, where
		$\mu_{4,\max}\triangleq\max_j\mu_{4,j}$, and it is therefore negligible
		relative to the coherent mainlobe for high $MN$.}}

\begin{Pro}
	\label{prop:MF_BR}
Under the dominant-mainlobe approximation, the surrogate-optimal MF power allocation of BS-$0$ is
	\begin{equation}
		\label{eq:MF_ramped_WF}
		\vspace{-1mm}
		P_n^{(0)\star}
		=
		\left[
		\frac{
			\lambda_m-\frac{1}{N}\sum_{\ell\neq 0}\Xi_{0,\ell}P_n^{(\ell)}
		}{
			2S^{(0)}_q b_n^{(0)}
		}
		\right]_+^{P_{\max}},
		 \forall n,
		 \vspace{-1mm}
	\end{equation}
	where $[x]_{+}^{P_{\max}} \triangleq \min\{P_{\max},\max\{x,0\}\}$, and $\lambda_m$ is chosen such that
	$\sum_{n=0}^{N-1}P_n^{(0)\star}=N P_{\mathrm{ave}}^{(0)}$.
\end{Pro}

\begin{proof}
	Under the equality sum-power constraint, the noise term in the
	MF SINR denominator is constant and can be omitted without affecting
	the optimizer. Therefore, the MF power-allocation subproblem is the convex program
	\[
	\vspace{-1mm}
	\min_{\mathbf P^{(0)}\succeq \mathbf 0}
	\;
	S^{(0)}_q\sum_{n=0}^{N-1} b_n^{(0)} \big(P_n^{(0)}\big)^2
	+
	\frac{1}{N}\sum_{\ell\neq 0}\Xi_{0,\ell}\sum_{n=0}^{N-1}P_n^{(0)}P_n^{(\ell)}
	\vspace{-1mm}
	\]
	subject to
	$
	\sum_{n=0}^{N-1}P_n^{(0)}=N P_{\mathrm{ave}}^{(0)}, P_n^{(0)}\ge 0.
	$
	The Lagrangian is
	$
	\mathcal L
	=
	S^{(0)}_q\sum_{n=0}^{N-1} b_n^{(0)} \big(P_n^{(0)}\big)^2
	+
	\frac{1}{N}\sum_{\ell\neq 0}\Xi_{0,\ell}\sum_{n=0}^{N-1}P_n^{(0)}P_n^{(\ell)}
	-\lambda_m\!\left(\sum_{n=0}^{N-1}P_n^{(0)}-N P_{\mathrm{ave}}^{(0)}\right)
	-\sum_{n=0}^{N-1}\eta_n P_n^{(0)}.
	$
	The Karush--Kuhn--Tucker (KKT)  stationarity condition for each $n$ is
	\begin{equation}
		\vspace{-1mm}
		2S^{(0)}_q b_n^{(0)} P_n^{(0)}
		+
		\frac{1}{N}\sum_{\ell\neq 0}\Xi_{0,\ell}P_n^{(\ell)}
		-\lambda_m-\eta_n=0.
		\vspace{-1mm}
	\end{equation}
	Adding the upper-bound multiplier for \(P_n^{(0)}\le P_{\max}\)
	gives the same KKT stationarity condition with clipping at
	\(P_{\max}\). Together with complementary slackness, this yields
	\eqref{eq:MF_ramped_WF}.
\end{proof}

According to Proposition~\ref{prop:MF_BR}, a higher $b_n^{(0)}$, induced by a larger $\mu_{4,n}^{(0)}$, suppresses $P_n^{(0)}$ because AF sidelobes are more heavily penalized on that subcarrier. The multiplier $\lambda_m$ can be found by one-dimensional bisection.

\subsubsection{Optimal Power Allocation for RF}

The RF SINR denominator at BS-$0$ is given by
\begin{equation}
	\label{eq:RF_denominator_struct}
	\vspace{-1mm}
	\mathcal D^{\mathrm{RF}}\big(\mathbf P^{(0)};\{\mathbf P^{(\ell)}\}\big)
	=\sum_{n=0}^{N-1}\frac{\mu^{(0)}_{-2,n}}{P^{(0)}_n}
	\Big(\sigma_z^2+\sum_{\ell\ne 0}\Xi_{0,\ell}P^{(\ell)}_n\Big).
	\vspace{-1mm}
\end{equation}
Let us now define the nonnegative weights
\begin{equation}
	\label{eq:RF_w_n}
	\vspace{-1mm}
	w^{(0)}_n \ \triangleq\ \mu^{(0)}_{-2,n}
	\Big(\sigma_z^2+\sum_{\ell \ne 0}\Xi_{0,\ell}P^{(\ell)}_n\Big).
	\vspace{-1mm}
\end{equation}
Then, given $\{P^{(\ell)}_n\}_{\ell \ne 0}$, minimizing
$\sum_n \frac{w^{(0)}_n}{P^{(0)}_n}$ under
\eqref{eq:avg_power_constraint} is strictly convex and admits a
closed-form solution.

\begin{Pro}
	\label{prop:RF_BR}
	The optimal RF power allocation is
	\begin{equation}
		\label{eq:RF_sqrt_rule_Pmax}
		\vspace{-1mm}
		P_n^{(0)\star}
		=
		\left[
		c_{\mathrm{RF}}\sqrt{w_n^{(0)}}
		\right]_{+}^{P_{\max}},
		\quad \forall n,
		\vspace{-1mm}
	\end{equation}
	where \(c_{\mathrm{RF}}>0\) is chosen to satisfy
	$
		\sum_{n=0}^{N-1}
		P_n^{(0)\star}
		=
		NP_{\mathrm{ave}}^{(0)}.
	$
	If \(P_{\max}\) is inactive, then
	$
		P_n^{(0)\star}
		=
		\frac{N P_{\mathrm{ave}}^{(0)}\sqrt{w_n^{(0)}}}
		{\displaystyle\sum_{k=0}^{N-1}\sqrt{w_k^{(0)}}},
		\quad \forall n.
	$
\end{Pro}

\begin{proof}
	The Lagrangian is
	\[
	\vspace{-1mm}
	\mathcal L=\sum_n \frac{w^{(0)}_n}{P^{(0)}_n}
	+\lambda_m\Big(\sum_n P^{(0)}_n-N P_{\mathrm{ave}}^{(0)}\Big)
	-\sum_n \eta_n P^{(0)}_n.
	\vspace{-1mm}
	\]
	Stationarity yields
	$-\frac{w^{(0)}_n}{(P^{(0)}_n)^2}+\lambda_m-\eta_n=0$.
	At the optimum $P^{(0)}_n>0$ (otherwise the objective diverges),
	so $\eta_n=0$ and $P^{(0)}_n=\sqrt{\tfrac{w^{(0)}_n}{\lambda_m}}$.
	With the upper-bound KKT multiplier for \(P_n^{(0)}\le P_{\max}\),
	the unconstrained square-root solution is clipped at \(P_{\max}\).
	Equivalently,
	\(P_n^{(0)\star}=[c_{\rm RF}\sqrt{w_n^{(0)}}]_{+}^{P_{\max}}\),
	where \(c_{\rm RF}\) is chosen to satisfy the sum-power constraint.
\end{proof}

This square-root rule shows that RF compensates tones associated with higher
reciprocal-filter weights, in contrast to the MF interference-avoidance
behavior summarized in Table~\ref{tab:sensing_only_closed_form}.

\begin{table}[t]
	\centering
	\caption{Sensing-only optimal power allocation under the DD-domain
		denominator-minimization surrogate.}
	\vspace{0mm}
	\label{tab:sensing_only_closed_form}
	\scriptsize
	\setlength{\tabcolsep}{3.5pt}
	\renewcommand{\arraystretch}{1.35}
	\begin{tabular}{
			>{\centering\arraybackslash}m{0.55cm}
			>{\centering\arraybackslash}m{3.35cm}
			>{\raggedright\arraybackslash}m{3.95cm}
		}
		\toprule
		Filter & Optimal power allocation 
		$\displaystyle
		P_n^{(0)}$ & Design implication \\
		\midrule
		MF
		&
		$
		\displaystyle
		\left[
		\frac{
			\lambda_m-\frac{1}{N}\sum_{\ell\neq 0}\Xi_{0,\ell}P_n^{(\ell)}
		}{
			2S^{(0)}_q b_n^{(0)}
		}
		\right]_+
		$
		&
		\emph{Interference avoidance:} larger power-overlap interference
		leads to lower allocated power.
		\\
		\midrule
		RF
		&
		$
		\frac{
			NP_{\rm ave}^{(0)}\sqrt{w_n^{\rm RF}}
		}{
			\sum_{k=0}^{N-1}\sqrt{w_k^{\rm RF}}
		}
		$
		&
		\emph{Interference compensation:} larger RF vulnerability leads to
		higher allocated power.
		\\
		\bottomrule
		\vspace{0.25mm}
	\end{tabular}
	\parbox{\linewidth}{\footnotesize \textit{Note:}
		The table highlights the uncapped closed-form rules. With an active
		peak-power constraint, the same structures are clipped by
		\([x]_{+}^{P_{\max}}\).}
\end{table}

\subsection{Tractable Denominator-Based Formulation}

{{We next derive tractable denominator-based formulations. For MF, this
		corresponds to the dominant-mainlobe approximation adopted in
		Proposition~\ref{prop:MF_BR}, where the power-dependent variation of the denominator is
		optimized, while the coherent mainlobe term is treated as fixed under the
		sum-power constraint.}} For RF, maximizing the SINR is exactly equivalent to
minimizing its denominator for a fixed desired path gain. Using
$\gamma_n^{\rm c}=g_nP_n^{(0)}$, the minimum power required by the communication
mode $j$ on subcarrier $n$ is $
	P_{n,j}^{\min}\triangleq \frac{\Gamma_j}{g_n}.$ Let us now introduce the auxiliary variable
\begin{equation}
	\label{eq:kappa_nj_def}
	\vspace{-1mm}
	\kappa_{n,j}\triangleq u_{n,j}P_n^{(0)}.
	\vspace{-1mm}
\end{equation}
Then, under \eqref{P1:select},
$
	P_n^{(0)}=\sum_{j=0}^{J}\kappa_{n,j},
$
and
$
	\label{eq:mu4P2_persp}
	\mu_{4,n}^{(0)}\big(P_n^{(0)}\big)^2
	=
	\sum_{j=0}^{J}
	\mu_{4,j}\frac{\kappa_{n,j}^2}{u_{n,j}},
$
where the auxiliary term is interpreted in the closed-convex sense, with $\kappa_{n,j}=0$, when $u_{n,j}=0$.

Using \eqref{eq:MF_denominator_reduced} and the fixed sum-power constraint, the finite-alphabet MF denominator-minimization surrogate becomes
\begin{subequations}\label{prob:P_MF_QOL}
	\begin{align}
		\min_{\bm{U},\bm{\kappa}}\quad
		&
		\sum_{n=0}^{N-1}\sum_{j=0}^{J}
		a_{q,j}^{\mathrm{MF}}
		\frac{\kappa_{n,j}^2}{u_{n,j}}
		+
		\sum_{n=0}^{N-1}\sum_{j=0}^{J}
		\lambda_n^{\mathrm{MF}}\kappa_{n,j}
		\\
		\text{s.t.}\quad
		&
		\frac{1}{N}\sum\nolimits_{n=0}^{N-1}\sum\nolimits_{j=1}^{J}u_{n,j} r_j
		\ge R_{\min},
		\\
		&
		\frac{1}{N}\sum\nolimits_{n=0}^{N-1}\sum\nolimits_{j=0}^{J}\kappa_{n,j}
		= P_{\mathrm{ave}}^{(0)},
		\kappa_{n,j}\ge 0,
		\\
		&
		u_{n,j}P_{n,j}^{\min}
		\le \kappa_{n,j}
		\le u_{n,j}P_{\max},
		 \forall n,\ j,
		\\
		&
		\sum\nolimits_{j=0}^{J}u_{n,j}=1,
		u_{n,j}\in\{0,1\}, \forall n,\ j,
	\end{align}
\end{subequations}
where 
$
a_{q,j}^{\mathrm{MF}}
\triangleq
S_q^{(0)}
\frac{
	M\mu_{4,j}-\frac{\mu_{4,j}-1}{N}
}{NM-1}
$
and
$
\lambda_n^{\mathrm{MF}}
\triangleq
\frac{1}{N}\sum_{\ell\neq 0}\Xi_{0,\ell}P_n^{(\ell)}.
$
Problem \eqref{prob:P_MF_QOL} is a mixed-integer convex quadratic-over-linear reformulation of the MF denominator-minimization surrogate.

For RF, maximizing \eqref{eq:SINR_RF_final_DD} is equivalent to minimizing its denominator. Under \eqref{P1:select}, we have
\begin{equation}
	\label{eq:mu_minus2_over_P_RF}
	\vspace{-1mm}
	\frac{\mu_{-2,n}^{(0)}}{P_n^{(0)}}
	=
	\sum_{j=0}^{J}
	\mu_{-2,j}\frac{u_{n,j}^2}{\kappa_{n,j}} .
	\vspace{-1mm}
\end{equation}
Let us define
$
	c_{n,j}^{\mathrm{RF}}
	\triangleq
	\mu_{-2,j}
	\left(
	\sigma_z^2+
	\sum_{\ell=1}^{L-1}\Xi_{0,\ell}P_n^{(\ell)}
	\right).
$
Then the finite-alphabet RF denominator-minimization problem is
\begin{subequations}\label{prob:P_RF_QOL}
	\begin{align}
		\vspace{-1mm}
		\min_{\bm{U},\bm{\kappa}}\quad
		&
		\sum\nolimits_{n=0}^{N-1}\sum\nolimits_{j=0}^{J}
		c_{n,j}^{\mathrm{RF}}
		\frac{u_{n,j}^2}{\kappa_{n,j}}
		\\
		\text{s.t.}\quad
		&
		\frac{1}{N}\sum\nolimits_{n=0}^{N-1}\sum\nolimits_{j=1}^{J}u_{n,j} r_j
		\ge R_{\min},
		\\
		&
		\frac{1}{N}\sum\nolimits_{n=0}^{N-1}\sum\nolimits_{j=0}^{J}\kappa_{n,j}
		= P_{\mathrm{ave}}^{(0)},
		\kappa_{n,j}\ge 0,
		\\
		&
		u_{n,j}P_{n,j}^{\min}
		\le \kappa_{n,j}
		\le u_{n,j}P_{\max}, \forall n,\ j,
		\\
		&
		\sum\nolimits_{j=0}^{J}u_{n,j}=1,
		u_{n,j}\in\{0,1\}, \forall n,\ j.
	\end{align}
\end{subequations}
Problem \eqref{prob:P_RF_QOL} is a mixed-integer convex quadratic-over-linear reformulation of the RF denominator-minimization problem.

\subsection{Inter-Cell Spectrum-Overlap Coordination}
\label{sec:multi_bs_power}
{We next study inter-cell spectrum coordination through a homogeneous
	twin-BS benchmark. To isolate the effect of spectrum-overlap cardinality,
	we restrict attention to a BS-symmetric allocation subclass, in which
	the two BSs use identical powers on shared tones and symmetric power
	profiles on equally sized exclusive-tone sets. Under this benchmark,
	we adopt a worst-BS fairness criterion to reveal whether spectrum
	sharing or orthogonalization is preferred under the receiver-induced
	interference structure. For clarity, per-tone peak-power limits are omitted.}


\subsubsection{MF: Full-Band Sharing versus Spectrum Orthogonalization}
\label{subsec:MF_multiBS}

For MF, the dominant coherent mainlobe term
$\frac{M}{N}(\sum_n P_n^{(\ell)})^2$ is fixed by the per-BS
sum-power constraint. Spectral reshaping therefore mainly changes
the denominator of the SINR expression in (\ref{SINR_derived_expression}) through two terms: the AF sidelobes quadratic
penalty and the inter-cell power-overlap penalty. We thus study the
following min--max denominator surrogate for two co-channel BSs,
$\ell\in\{0,1\}$:
\begin{equation}
	\label{eq:minmax_den_def}
	\begin{aligned}
		\vspace{-1mm}
		\min_{\{P^{(0)}_n,P^{(1)}_n\ge 0\}}
		\quad & \max\{D_0,D_1\} \\
		\text{s.t.}\quad
		& \sum_{n=0}^{N-1}P^{(0)}_n=NP_{\mathrm{ave}},
		\sum_{n=0}^{N-1}P^{(1)}_n=NP_{\mathrm{ave}},
		\vspace{-1mm}
	\end{aligned}
\end{equation}
with
$
D_\ell \triangleq \sum_{n=0}^{N-1}
\Big(
S_q^{(\ell)}b_n^{(\ell)}(P^{(\ell)}_n)^2
+
\frac{1}{N}\Xi_{\ell,\bar \ell}P^{(\ell)}_n P^{(\bar \ell)}_n
\Big),
$
where $\bar \ell\neq \ell$, $S_q^{(\ell)}b_n^{(\ell)}>0$, and $\Xi_{\ell,\bar \ell}\ge 0$.

Assume the following homogeneous symmetric structure:
$
S_q^{(0)}=S_q^{(1)}\equiv S_q,
b_n^{(0)}=b_n^{(1)}\equiv b,
\Xi_{0,1}=\Xi_{1,0}\equiv \Xi,
$
with
$
b \triangleq
\frac{
	M\mu_4-\frac{\mu_4-1}{N}
}{NM-1},
$
where
$
\mu^{(0)}_{4,n}=\mu^{(1)}_{4,n}\equiv \mu_4.
$
Let $N_{\mathrm{ov}}$ denote the number of overlap tones and define
\[
\mathcal F \triangleq \{\,N_{\mathrm{ov}}\in\{0,1,\dots,N\}: N-N_{\mathrm{ov}}\ \text{is even}\,\}.
\]
For each $N_{\mathrm{ov}}\in\mathcal F$, the tones are partitioned into the shared set $\mathcal O$ with $|\mathcal O|=N_{\mathrm{ov}}$ and the exclusive sets $\mathcal E_0,\mathcal E_1$ with
$  
|\mathcal E_0|=|\mathcal E_1|=\frac{N-N_{\mathrm{ov}}}{2}.
$
\begin{Pro}
	\label{thm:twoBS_endpoint}
	Within the BS-symmetric overlap-cardinality subclass, the optimized MF
	denominator is minimized by
	$N_{\mathrm{ov}}^\star=\min\mathcal F$ if
	$\frac{\Xi}{N}>S_qb$, and by
	$N_{\mathrm{ov}}^\star=N$ if
	$\frac{\Xi}{N}<S_qb$.
\end{Pro}
\begin{proof}
	Please refer to Appendix C.
\end{proof}

\begin{remark}
	{Proposition~\ref{thm:twoBS_endpoint} reveals that increasing the overlap
		introduces additional inter-cell interference, but also allows each BS
		to spread its fixed transmit power over more tones, thereby reducing
		the quadratic self-sidelobe penalty. When $\Xi/N<S_qb$, this reduction
		dominates the added overlap interference, and full-band sharing yields
		a lower total DD-domain denominator.}
\end{remark}

\subsubsection{Masked RF: Minimum-Overlap Tendency}
\label{subsec:RF_multiBS}

{{The preceding RF SINR and power-allocation results are derived for the full-band RF implementation, where all sensing tones remain active. To additionally characterize the effect of hard spectrum coordination, we now consider a masked-RF extension in this section. In this variant, a BS applies the reciprocal weight only on its active sensing tones and sets the RF weight to zero on inactive tones.}} Consequently, the RF mainlobe gain and denominator are both evaluated over the active sensing set, rather than over all $N$ subcarriers.\footnote{{Additional range sidelobes or ambiguity peaks induced by the active-tone pattern in masked RF are not considered and can be mitigated through tone-pattern design or weighted processing.}}

Using the same symmetric partition of $\{\mathcal O,\mathcal E_0,\mathcal E_1\}$
defined above, each BS is active on
$|\mathcal O\cup\mathcal E_\ell|=(N+N_{\mathrm{ov}})/2$ tones.
Under the per-BS sum-power constraint
$
\sum_{n\in \mathcal O\cup \mathcal E_\ell} P_n^{(\ell)} = N P_{\mathrm{ave}},
 \ell\in\{0,1\},
$
the RF mainlobe gain on the matched DD bin is
$
\frac{M}{N}\Big(\frac{N+N_{\mathrm{ov}}}{2}\Big)^2.
$
Therefore, in contrast to MF, the RF numerator does depend on the overlap cardinality, but only through $N_{\mathrm{ov}}$, rather than the specific tone indices.

Under the homogeneous abstraction, for fixed $N_{\mathrm{ov}}$ we adopt the following min--max RF-denominator surrogate:
\begin{equation}
	\label{eq:RF_minmax_den_def}
	\begin{aligned}
		\vspace{-1mm}
		\min_{\{P_n^{(0)},P_n^{(1)}\ge 0\}}
		\quad & \max\{D_0,D_1\} \\
		\text{s.t.}\quad
		& \sum_{n\in \mathcal O\cup \mathcal E_0} P_n^{(0)} = N P_{\mathrm{ave}},\\
		& \sum_{n\in \mathcal O\cup \mathcal E_1} P_n^{(1)} = N P_{\mathrm{ave}},\\
		& P_n^{(0)}=0,\ \forall n\in \mathcal E_1,
		P_n^{(1)}=0,\ \forall n\in \mathcal E_0,
		\vspace{-1mm}
	\end{aligned}
\end{equation}
with
$
D_\ell \triangleq
\sum_{n\in \mathcal O\cup \mathcal E_\ell}
\frac{\mu_{-2,n}^{(\ell)}\sigma_z^2}{P_n^{(\ell)}}
+
\sum_{n\in \mathcal O}
\mu_{-2,n}^{(\ell)}\Xi_{\ell,\bar \ell}\frac{P_n^{(\bar \ell)}}{P_n^{(\ell)}},
$
where $\bar \ell\neq \ell$, $\mu_{-2,n}^{(\ell)}\sigma_z^2>0$, and $\Xi_{\ell,\bar \ell}\ge 0$.

Then, we assume the following homogeneous symmetric structure:
$
\mu_{-2,n}^{(0)}=\mu_{-2,n}^{(1)}\equiv \mu_{-2},
\Xi_{0,1}=\Xi_{1,0}\equiv \Xi, \forall n,
$
with $\mu_{-2}\sigma_z^2>0$ and $\Xi\ge 0$.
For RF, activating fewer subcarriers reduces not only the interference and noise terms, but also the coherent mainlobe gain of the effective channel. Specifically, if each BS is active on
$
K \triangleq \frac{N+N_{\mathrm{ov}}}{2}
$
subcarriers, then the coherently combined target-signal power scales as $K^2$. Therefore, for RF, the overlap pattern must be optimized with respect to the full SINR expression, rather than only its denominator.

\begin{Pro}
	\label{thm:RF_twoBS_endpoint}
	For masked RF with $\Xi>0$, the max-min worst-case RF SINR is maximized
	by $
		N_{\mathrm{ov}}^\star=\min\mathcal F,$
	where $N_{\mathrm{ov}}^\star=0$ for even $N$ and
	$N_{\mathrm{ov}}^\star=1$ for odd $N$. 
\end{Pro}

\begin{proof}
	Please refer to Appendix D.
\end{proof}

\begin{remark}
	In this symmetric overlap model, MF and masked RF lead to different
	coordination principles. MF exhibits a thresholding behavior: spectral
	orthogonalization is preferred only when the normalized inter-cell
	coupling of $\Xi/N$ is higher than the AF sidelobe coefficient $S_qb$;
	otherwise, full-band overlap can be more efficient. By contrast,
	masked RF favours the minimum feasible overlap whenever inter-cell
	coupling is present, because each shared active tone introduces a
	ratio-type interference penalty. 
\end{remark}

\section{Beyond-CP Inter-Cell Leakage and Effective Spectrum}
\label{sec:prop_delay_interference}

The preceding sections assume that the effective inter-cell delays remain
within the CP, so that the interference from BS-$\ell$ is separable across
subcarriers and described by its nominal spectrum $\{P_n^{(\ell)}\}$.
In dense OFDM-ISAC networks, however, some inter-cell components may exceed
the CP and redistribute the interfering power across fast Fourier transform (FFT)  bins. This section
captures this effect through a delay-distorted effective interference
spectrum and shows how the previous MF/RF interference terms should be
modified. To isolate propagation-delay effects, we neglect Doppler and carrier-frequency offset (CFO) on
the quasi-static inter-cell coupling links.

\subsection{Effective Interference Spectrum under CP Violation}

For a path impinging from BS-$\ell$ to BS-0 with delay $\tau$, we define the
per-subcarrier effective interference power after CP removal and FFT as
\begin{equation}
	\label{eq:eta_delay_def}
	\vspace{-1mm}
	\eta_n^{(\ell)}(\tau)
	\triangleq
	\mathbb E\!\left[
	\left|
	\big[
	\mathcal U^{\mathrm f}_{\tau,0}
	(\mathbf X_m^{(\ell)},\mathbf X_{m-1}^{(\ell)})
	\big]_n
	\right|^2
	\right].
	\vspace{-1mm}
\end{equation}
If $0\le \tau\le N_{\mathrm{CP}}$, then
$
	\eta_n^{(\ell)}(\tau)=P_n^{(\ell)} .
$
If $\tau>N_{\mathrm{CP}}$, let us define the integer excess delay
$d(\tau)\triangleq \tau-N_{\mathrm{CP}}$ and assume
$d(\tau)\in\{1,\ldots,N-1\}$, i.e., at most one-symbol ISI is
induced.\footnote{Here, $\tau_q^{(\ell,0)}$ is the effective inter-cell delay measured relative to
	the FFT window of BS-0. The inter-cell delay-power profile
	$\{|\alpha_q^{(\ell,0)}|^2,\tau_q^{(\ell,0)}\}$ is treated as slowly varying
	coupling information, which can be obtained from reference-signal measurements, or long-term interference
	measurements.
}
Then
$\eta_n^{(\ell)}(\tau)
		=
		\sum_{r=0}^{N-1} P_r^{(\ell)}
		\Big(
		\big|[\mathbf F_N\mathbf C_0(d(\tau))\mathbf F_N^H]_{n,r}\big|^2
		+
		\big|[\mathbf F_N\mathbf C_1(d(\tau))\mathbf F_N^H]_{n,r}\big|^2
		\Big).$
Thus, CP `violation' maps the transmitted power on subcarrier $r$ to all
receive subcarriers through the delay-induced coupling matrices.

Consistent with the uncorrelated-path model used in the DD-domain SINR
analysis, the path-level interference powers are added noncoherently.
For BS-$\ell$, the effective interference spectrum at BS-0 is formulated as
\begin{equation}
	\label{eq:Peff_delay_def}
	\vspace{-1mm}
	P_{\mathrm{eff}}^{(\ell)}[n]
	\triangleq
	\sum\nolimits_{q=0}^{Q_\ell}
	|\alpha_q^{(\ell,0)}|^2\,
	\eta_n^{(\ell)}\!\big(\tau_q^{(\ell,0)}\big).
	\vspace{-1mm}
\end{equation}
By construction, when all inter-cell paths satisfy $\tau_q^{(\ell,0)}\le N_{\mathrm{CP}}$, \eqref{eq:Peff_delay_def} reduces to
\begin{equation}
	\label{eq:Peff_delay_withinCP}
	\vspace{-1mm}
	P_{\mathrm{eff}}^{(\ell)}[n]
	=
	\Big(\sum\nolimits_{q=0}^{Q_\ell}|\alpha_q^{(\ell,0)}|^2\Big) P_n^{(\ell)},
	\vspace{-1mm}
\end{equation}
which leads to the within-CP delay model.

Moreover, since $\mathbf F_N$ is unitary and
$\mathbf C_0(d),\mathbf C_1(d)$ partition the current/previous-symbol
contributions, the total average interference power is preserved:
\begin{equation}
	\label{eq:Peff_power_conservation}
	\vspace{-1mm}
	\sum\nolimits_{n=0}^{N-1}P_{\mathrm{eff}}^{(\ell)}[n]
	=
	\left(\sum\nolimits_{q=0}^{Q_\ell}|\alpha_q^{(\ell,0)}|^2\right)
	\sum\nolimits_{r=0}^{N-1}P_r^{(\ell)}.
	\vspace{-1mm}
\end{equation}
Thus, for a given aggregate inter-cell path gain, beyond-CP delays do not
increase the total interference energy. They only redistribute it across
subcarriers.

\subsection{Implications for Spectrum Coordination}

Upon replacing the nominal inter-cell spectrum by
$P_{\mathrm{eff}}^{(\ell)}[n]$, the MF multi-cell EISL term becomes
\begin{equation}
	\label{eq:EISL_MF_delay}
	\vspace{-1mm}
	\mathrm{EISL}^{\mathrm{MF}}_{0,\mathrm{multi}}
	=
	\frac{1}{N}\sum_{\ell=1}^{L-1}\sum_{n=0}^{N-1}
	P_n^{(0)}P_{\mathrm{eff}}^{(\ell)}[n].
	\vspace{-1mm}
\end{equation}
Under RF, the scaled denominator term in (\ref{eq:SINR_RF_final_DD}) becomes
\begin{equation}
	\label{eq:RF_den_delay}
	\vspace{-1mm}
	\mathcal D^{\mathrm{RF}}
	=
	\sum_{n=0}^{N-1}
	\frac{\mu_{-2,n}^{(0)}}{P_n^{(0)}}
	\Big(
	\sigma_z^2+\sum_{\ell=1}^{L-1}P_{\mathrm{eff}}^{(\ell)}[n]
	\Big).
	\vspace{-1mm}
\end{equation}
Thus, RF sees the same effective spectrum, further weighted by the
reciprocal-filter amplification factor $\mu_{-2,n}^{(0)}/P_n^{(0)}$.

The above expressions imply a simple replacement principle: the MF and
RF power-allocation structures derived under the within-CP delay model remain
valid after replacing the nominal inter-cell spectrum by the effective
spectrum $P_{\mathrm{eff}}^{(\ell)}[n]$. Hence, CP violation does not
change the qualitative MF/RF design laws. MF still avoids tones associated with
high effective overlap, whereas RF still compensates tones having high
delay-aware reciprocal-filter weights.

This replacement also clarifies the robustness of spectrum coordination.
Within the CP, nominal spectral orthogonalization can suppress overlap.
Beyond the CP, however, a nominally disjoint interferer can still leak
power into protected tones, so strict orthogonalization becomes less
effective. For a flat full-band interferer, the effective spectrum remains
flat; hence beyond-CP leakage mainly matters for frequency-division multiplexing (FDM), partial-band
activation, or strongly shaped spectra.

\subsection{Delay-Induced Leakage Kernel and Bound}
\label{subsec:single_tone_leakage}

We finally quantify the severity of delay-induced spectral redistribution
using a single-tone interferer, which characterizes the elementary leakage response underlying the effective interference spectrum. Consider a single beyond-CP path from BS-$\ell$
to BS-0 with integer excess delay
$
d=d(\tau)\in \{1,\ldots,N-1\}.
$
Assume that BS-$\ell$ transmits only on subcarrier $r$, i.e.,
\begin{equation}
	\label{eq:single_tone_input}
	\vspace{-1mm}
	X_m^{(\ell)}[n]
	=
	\sqrt{P} s_m \delta[n-r],
	\vspace{-1mm}
\end{equation}
where $\{s_m\}$ are zero-mean and independent across OFDM symbols, and $\mathbb E[|s_m|^2]=1$. Using
the beyond-CP FD operator, the FFT-domain contribution on
receive subcarrier $n$ is formulated as:
\begin{equation}
	\label{eq:single_tone_output}
	Y_m[n]
	=
	\sqrt{P} [\mathbf G_0(d)]_{n,r}s_m
	+
	\sqrt{P} [\mathbf G_1(d)]_{n,r}s_{m-1}.
\end{equation}
Since $s_m$ and $s_{m-1}$ are independent, the normalized leakage kernel is
\begin{equation}
	\label{eq:single_tone_kernel_def}
	H_d[\Delta]
	\triangleq
	|[\mathbf G_0(d)]_{n,r}|^2
	+
	|[\mathbf G_1(d)]_{n,r}|^2,
	\Delta\equiv n-r \pmod N,
\end{equation}
so that
$
\mathbb E[|Y_m[n]|^2]=P H_d[\Delta].
$
For the nominal FFT bin $\Delta=0$, direct evaluation gives
$
	H_d[0]
	=
	\frac{(N-d)^2+d^2}{N^2}.$
For $\Delta\neq 0$, the off-bin leakage is expressed as
$
	H_d[\Delta]
	=
	\frac{2}{N^2}
	\frac{\sin^2 \big(\pi d\Delta/N\big)}
	{\sin^2 \big(\pi \Delta/N\big)},
	\Delta\neq 0 .$
Since the total power is conserved, we have $\sum_{\Delta=0}^{N-1}H_d[\Delta]=1$.
Therefore, the fraction of single-tone power leaked from the nominal
FFT bin is
\begin{equation}
	\label{eq:Lambda_d_def}
	\vspace{-1mm}
	\Lambda(d)
	\triangleq
	1-H_d[0]
	=
	\frac{2d(N-d)}{N^2}
	\le \frac{1}{2}.
	\vspace{-1mm}
\end{equation}
Thus, under one-symbol CP violation, at most half of the power of a single
interfering tone can leave its nominal FFT bin. The excess delay $d$ does
not determine a finite number of affected subcarriers; instead, it controls
the Dirichlet-squared leakage shape, the null locations, and the total
leaked power fraction. This characterizes the leakage underlying the effective interference spectrum and explains why nominal spectral
orthogonalization can be degraded by beyond-CP delay-induced inter-cell leakage.

\section{Simulations}

Unless otherwise specified, the simulations consider an OFDM-ISAC network with \(N=64\) subcarriers and \(M=16\) OFDM symbols. The noise variance is set to \(\sigma_z^2=0.01\). For the adaptive constellation and power-control designs, the average and peak transmit powers of BS-0 are \(P_{\rm ave}^{(0)}=8\) and \(P_{\max}=16\), respectively. Each subcarrier selects a specific mode from a sensing-only waveform and seven communication constellations, namely QPSK, 8APSK, 16QAM, 16PSK, 32APSK, 64QAM, and 256QAM. 

\begin{figure}[t]
	\centering
	\includegraphics[width=7.5cm]{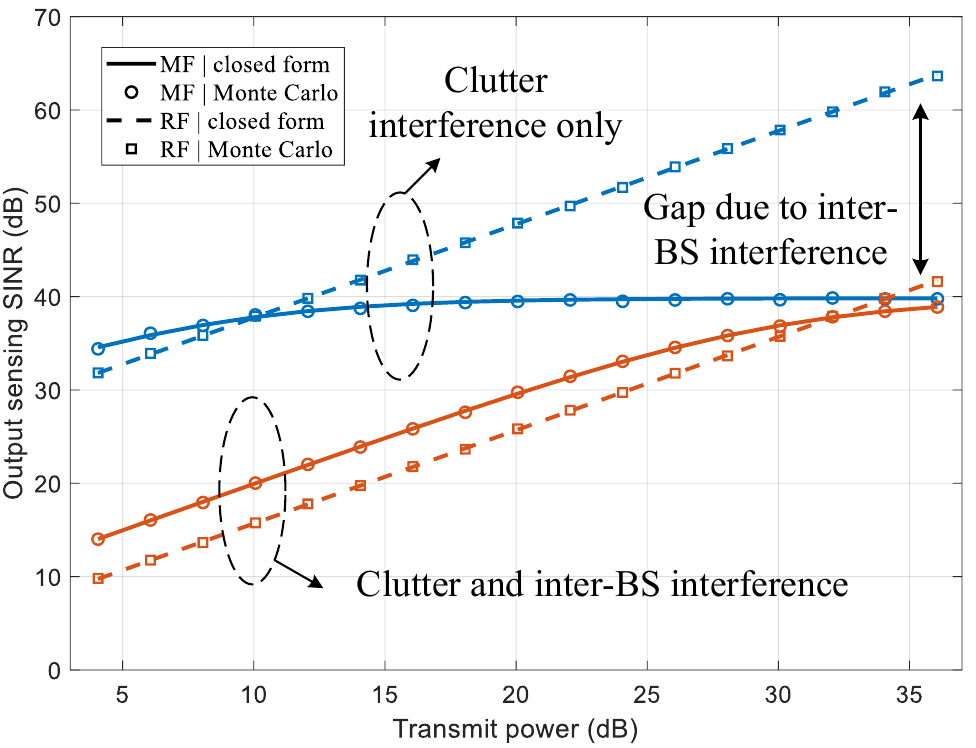}
	\vspace{-3mm}
	\caption{Analytical and Monte Carlo sensing SINR versus the average
	transmit power $P_{\rm ave}^{(0)}$ of BS-0.}
	\label{figure5}
\end{figure}

Fig.~\ref{figure5} validates the proposed SINR analysis by comparing the results from closed-form expressions to those of Monte Carlo simulations. The analytical curves and simulation markers closely match for both MF and RF, under both clutter-only and clutter-plus-inter-cell-interference scenarios. The results also show that inter-cell interference causes a pronounced SINR degradation, demonstrating that multi-cell coupling cannot be treated as negligible background noise. Moreover, MF and RF exhibit different power-scaling trends: MF is mainly limited by the clutter floor, whereas RF can achieve a steeper SINR increase when the dominant impairment is filtered noise. This confirms the need for receive-filter-aware interference management in multi-cell OFDM-ISAC.

\begin{figure}[t]
	\centering
	\includegraphics[width=7.5cm]{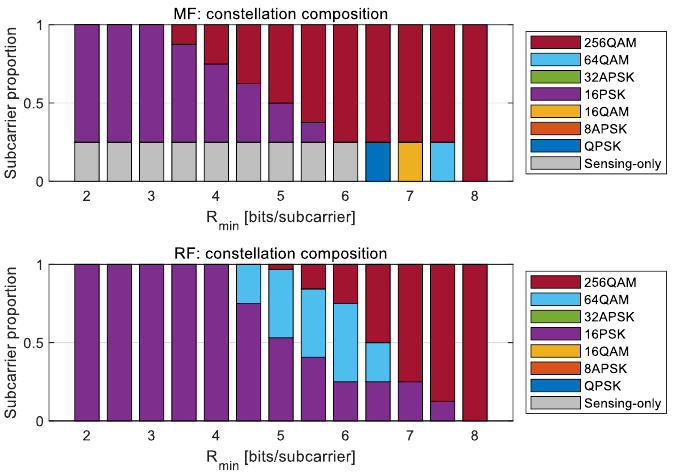}
	\vspace{-3mm}
	\caption{Constellation selection versus $R_{\min}$ constraints for MF and RF.}
	\label{figure8}
\end{figure}

The constellation-composition results reveal the impact of the communication requirement on the proposed BER-aware OFDM-ISAC resource allocation. As $R_{\min}$ increases, both MF and RF progressively shift the selected subcarriers from low-order constellations, mainly QPSK, to higher-order constellations such as 64QAM and 256QAM, which reflects the increasing communication payload pressure. Under MF, the mode allocation exhibits a more diversified mixture of QPSK, 8APSK, 16QAM, 64QAM, and 256QAM over the medium-payload region, indicating that MF can flexibly exploit amplitude-varying constellations, while balancing sensing interference and BER-induced power constraints. By contrast, RF tends to favor constant-envelope or more RF-compatible constellations, especially QPSK and 16PSK, over a wider range of moderate $R_{\min}$, before eventually switching to 256QAM when the payload constraint becomes stringent. This behavior is consistent with the RF structure, where the inverse-symbol weighting is more sensitive to low constellation amplitudes and therefore implicitly penalizes some amplitude-varying constellations. 

\begin{figure}[t]
	\centering
	\includegraphics[width=7.5cm]{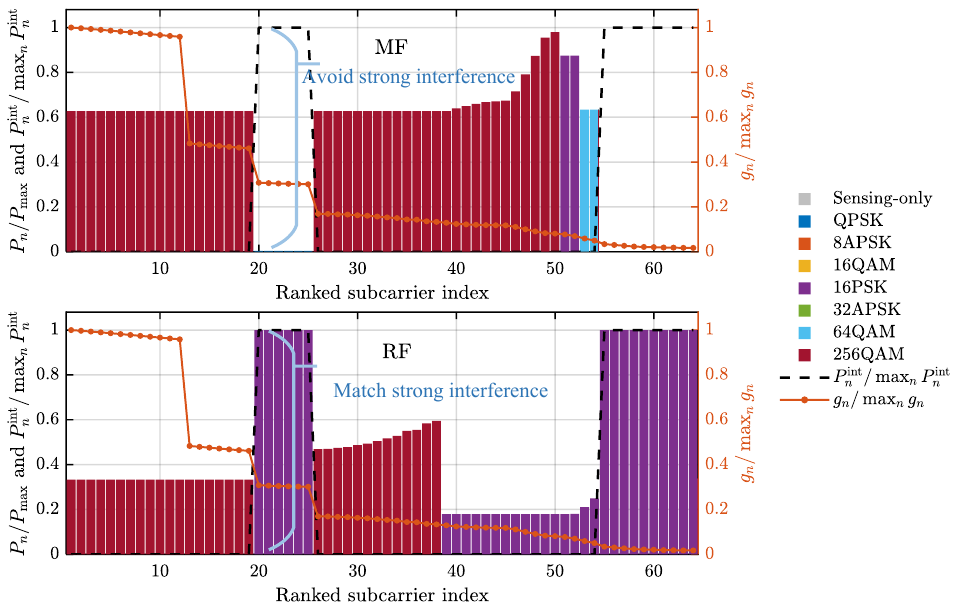}
	\vspace{-3mm}
	\caption{Optimized power and constellation allocation across subcarriers
		under $R_{\min}=6$ bits/channel use.}
	\label{figure10}
\end{figure}

Fig. \ref{figure10} shows the optimized constellation and power allocation under $R_{\min}=6$, with subcarriers sorted by the communication channel gain $g_n$. A clear filter-dependent behavior is observed. For MF, the allocated power tends to avoid subcarriers with strong inter-cell interference, and several highly interfered tones are assigned little or no power. This is consistent with the MF interference term $\sum_n P_n^{(0)}P_n^{\mathrm{int}}$, where reducing spectral overlap directly lowers the sensing interference floor. By contrast, RF allocates power in a way that more closely follows the interference profile. Since the RF denominator contains ratio-type terms such as $P_n^{\mathrm{int}}/P_n^{(0)}$, assigning more power to strongly interfered or vulnerable tones can reduce reciprocal-filter amplification. Thus, MF prefers interference avoidance, whereas RF tends to balance its power against the interference level.

\begin{figure}[t]
	\centering
	\includegraphics[width=7cm]{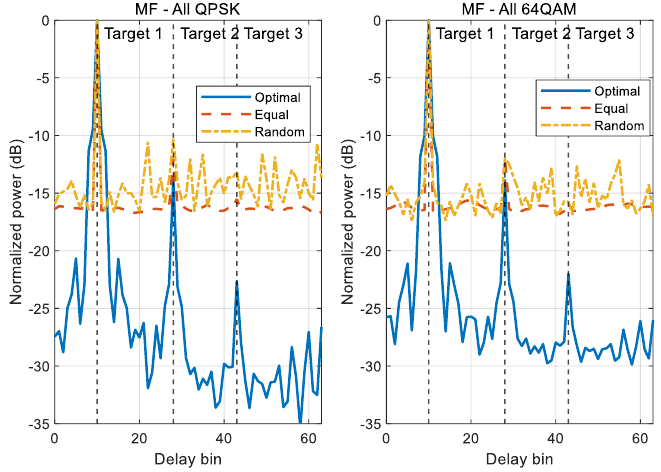}
	\vspace{-3mm}
	\caption{MF sensing profiles with optimized power allocation}
	\label{figure6a}
\end{figure}

\begin{figure}[t]
	\centering
	\includegraphics[width=7cm]{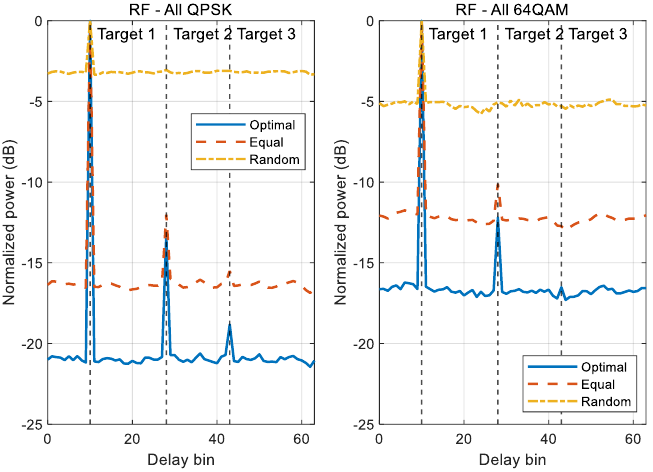}
	\vspace{-3mm}
	\caption{RF sensing profiles with optimized power allocation.}
	\label{figure6b}
\end{figure}

Fig.~\ref{figure6a} and Fig.~\ref{figure6b} compare the delay-domain sensing profiles of MF and RF under all-QPSK and all-64QAM signaling. The proposed optimal power allocation is benchmarked against equal-power and random-power baselines, and the vertical dashed lines mark the true delay bins of the three targets. The optimized allocation significantly lowers the sidelobe/interference floor for both receivers and yields clearer peaks at the target delays, especially for the weaker targets. 
Under MF, the improvement accrues from reducing the power-overlap-induced interference and the data-dependent sidelobe floor, which is consistent with the MF SINR in (\ref{SINR_derived_expression}) involving $\mu_{4,n}^{(0)}$ and $\sum_n P_n^{(0)}P_n^{(\ell)}$. Under RF, the all-QPSK case achieves a cleaner profile than all-64QAM because QPSK has a more favorable inverse-symbol-power behavior, whereas 64QAM is more sensitive to reciprocal-filter amplification through $\mu_{-2,n}^{(0)}$ of equation (\ref{eq:SINR_RF_final_DD}). 
These results demonstrate that receive filtering, constellation format, and power allocation must be jointly designed for robust weak-target detection in multi-cell OFDM-ISAC systems.

\begin{figure}[t]
	\centering
	\includegraphics[width=7cm]{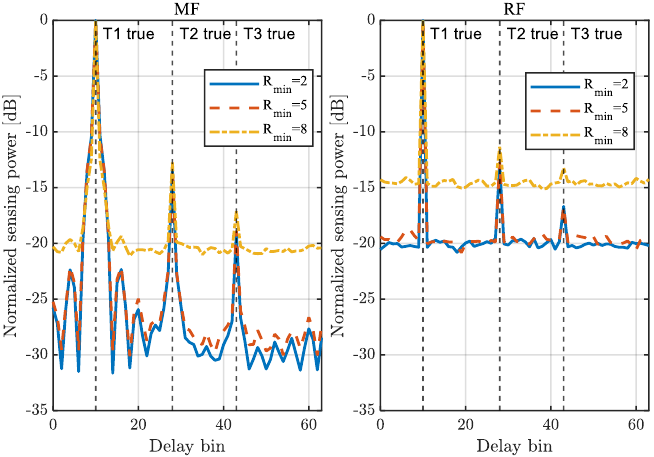}
	\vspace{-3mm}
	\caption{Impact of communication payload requirements on sensing profiles.}
	\label{figure7}
\end{figure}
{Fig.~\ref{figure7} compares the delay-domain sensing profiles obtained by
	jointly optimizing constellation selection and power allocation under
	different communication payload requirements.} For both MF and RF filters, the strongest target can be reliably identified under all $R_{\min}$ values considered. However, increasing $R_{\min}$ raises the residual sidelobe/interference floor and reduces the relative visibility of the weaker targets. This demonstrates that stricter communication QoS constraints limit the degrees of freedom available for sensing-oriented constellation and power optimization.
For MF, the degradation appears as a gradual increase in the sidelobe floor, which is consistent with the dependence of the MF EISL on the fourth-order constellation moment and the power-overlap term in (\ref{SINR_derived_expression}). For RF, the high-$R_{\min}$ regime leads to a more pronounced increase in the sidelobe and interference floor because RF is sensitive to inverse-symbol-power moments and reciprocal-filter amplification. These results confirm the proposed sensing--communication tradeoff: higher communication payload requirements may be achieved at the cost of reduced sensing dynamic range and weaker target detectability.

\begin{figure}[t]
	\centering
	\includegraphics[width=7cm]{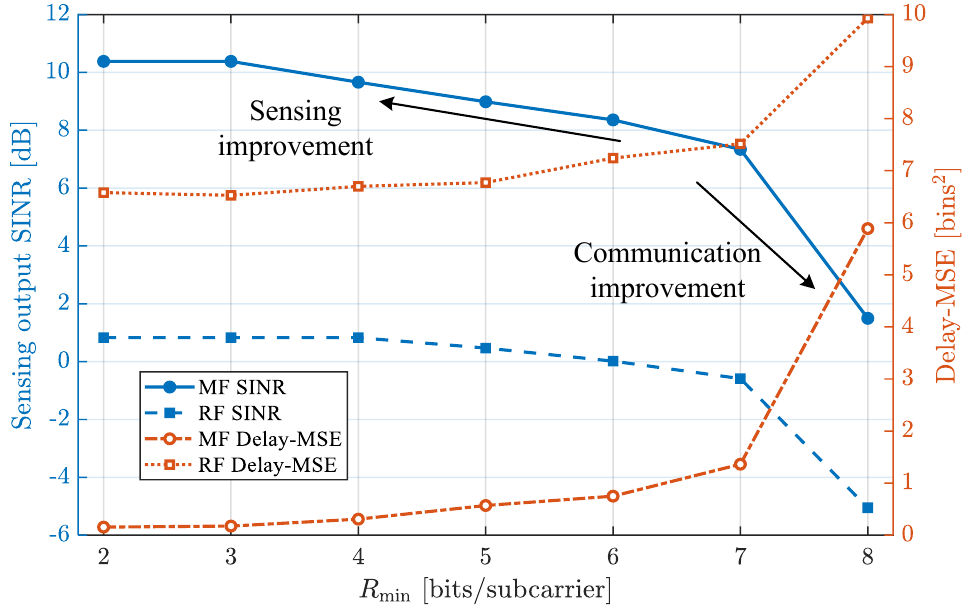}
	\vspace{-3mm}
	\caption{Sensing SINR and MSE versus communication payload requirement.}
	\label{figure9}
\end{figure}
The sensing-SINR results in Fig. \ref{figure9} illustrate the fundamental tradeoff between communication and sensing performance under the proposed BER-aware constellation and power allocation design. As the minimum bits per channel use requirement $R_{\min}$ increases, more subcarriers must employ higher-order constellations and satisfy stricter BER-induced power requirements, which gradually reduces the degrees of freedom available for sensing-oriented waveform shaping. Consequently, the sensing SINR of all targets decreases with $R_{\min}$, especially in the high-payload regime beyond approximately 6.5 bits/subcarrier. The delay-wise MSE exhibits the opposite trend, increasing with $R_{\min}$ as the sensing-oriented waveform degrees of freedom are reduced. These observations confirm that the proposed framework captures a controllable sensing–communication tradeoff: moving toward the lower-$R_{\min}$ region improves sensing, whereas increasing $R_{\min}$ enhances the communication throughput at the cost of sensing SINR reduction. 

\section{Conclusions}

A modulation- and receive-filter-aware framework was developed for multi-cell OFDM-ISAC under clutter and inter-cell interference. Closed-form sensing SINR expressions were derived for MF and RF receivers, revealing that MF is mainly governed by fourth-order constellation moments and power-overlap interference. By contrast, RF is dominated by inverse-symbol-power moments and ratio-type interference. Based on these expressions, sensing-oriented power allocation structures were obtained, showing that MF follows a ramped water-filling behavior, while RF obeys a square-root allocation rule. The joint finite-alphabet constellation selection and power allocation problem was further formulated under communication and power constraints. The analysis also characterized inter-cell spectrum coordination and beyond-CP delay-induced interference leakage, showing when spectral orthogonalization is preferred in the symmetric two-BS benchmark and how propagation-delay-induced leakage redistributes interference power across subcarriers. Numerical results validated the theoretical SINR expressions and demonstrated that constellation and power allocation can effectively reduce the DD-domain interference floor, improve weak-target visibility, and reveal the sensing--communication tradeoff in interference-limited multi-cell OFDM-ISAC systems.

\vspace{-1.5mm}

\section*{Appendix A: \textsc{Proof of Proposition \ref{MFSINR_Expression}}}
\vspace{-1.5mm}
To connect these quantities with the 1-D range kernels, we commence from the MF
self-kernel 
\begin{equation}
	\vspace{-1mm}
	\chi^{(0,0)}_{\mathrm{MF}}[k,p]
	= \frac{1}{\sqrt{NM}}\sum_{m,n}
	|X^{(0)}_m[n]|^2 e^{j\frac{2\pi}{N}nk}e^{-j\frac{2\pi}{M}m p},
	\vspace{-1mm}
\end{equation}
and recall that $X^{(0)}_m[n]=\sqrt{P^{(0)}_n}s^{(0)}_{n,m}$. 
$\mathbb E[|s^{(0)}_{n,m}|^2]=1$ and
$\mathbb E[|s^{(0)}_{n,m}|^4]=\mu^{(0)}_{4,n}$.  Straightforward expansion of
$|\chi^{(0,0)}_{\mathrm{MF}}[k,p]|^2$ and averaging with respect to the
independent symbols across $(n,m)$ yields
	\begin{align}
		\label{eq:MF_A_second}
		\vspace{-1mm}
		\mathbb E \big[|\chi^{(0,0)}_{\mathrm{MF}}[k,p]|^2\big]
		&= \frac{1}{NM}
		\left|\sum_{n=0}^{N-1}P^{(0)}_n e^{j\frac{2\pi}{N}nk}\right|^2
		\left|\sum_{m=0}^{M-1}e^{-j\frac{2\pi}{M}m p}\right|^2   \nonumber \\
		&\quad + \frac{1}{N}\sum_{n=0}^{N-1}
		\big(P^{(0)}_n\big)^2\big(\mu^{(0)}_{4,n}-1\big). 
		\vspace{-1mm}
	\end{align}
In \eqref{eq:MF_A_second}, for the main peak $[k,p]=(0,0)$, we have
\begin{equation}
	\label{eq:MF_A_00}
	\vspace{-1mm}
	\mathbb E \big[|\chi^{(0,0)}_{\mathrm{MF}}[0,0]|^2\big]
	\! = \! \frac{M}{N} \! \Big( \!\sum_{n=0}^{N-1}P^{(0)}_n\Big)^2
	\! +  \frac{1}{N}\sum_{n=0}^{N-1}\big(P^{(0)}_n\big)^2\big(\mu^{(0)}_{4,n}-1\big).
	\vspace{-1mm}
\end{equation}

For the inter-cell kernels $\chi^{(\ell,0)}_{\mathrm{MF}}[k,p]$, an analogous
calculation with $X^{(\ell)}_m[n]$ in place of $X^{(0)}_m[n]$ and using
the cross-correlation definition gives
\begin{equation}
	\label{eq:mf_chi_inter_second}
	\vspace{-1mm}
	\mathbb E \big[|\chi^{(\ell,0)}_{\mathrm{MF}}[k,p]|^2\big]
	= \frac{1}{N}\sum\nolimits_{n=0}^{N-1} P^{(0)}_n P^{(\ell)}_n.
	\vspace{-1mm}
\end{equation}	
The total energy of $\chi^{(\ell,0)}_{\mathrm{MF}}[k,p]$ over all Doppler bins is
governed by \eqref{eq:mf_chi_inter_second} via Parseval’s theorem.  

Averaging over random data symbols and path coefficients, $\mathrm{EISL}^{\mathrm{MF}}_q$ is defined as the symbol-averaged power of all undesired terms on the matched DD bin.
Upon collecting all contributions, we have the mean
{effective interference-plus-sidelobe level} (EISL) for the $q$-th DD bin
as
	\begin{align}
		\vspace{-1mm}
		\label{eq:EISL_MF_final_DD}
		\mathrm{EISL}^{\mathrm{MF}}_q
		&\triangleq \sum_{\substack{q'=0\\ q'\ne q}}^{Q_0}  \!  |\alpha_{q'}^{(0)}|^2 
		\mathbb E \Big[
		\big|\chi^{(0,0)}_{\mathrm{MF}}\big[\tau_q^{(0)} \! - \! \tau_{q'}^{(0)},\,
		\psi_q^{(0)} \! - \! \psi_{q'}^{(0)}\big]\big|^2
		\Big] \nonumber \\
		&\quad + \frac{1}{N}\sum_{\ell=1}^{L-1}\Big(\sum_{q'=0}^{Q_\ell}
		|\alpha_{q'}^{(\ell, 0)}|^2\Big) 
		\sum_{n=0}^{N-1} P^{(0)}_n P^{(\ell)}_n.
		\vspace{-1mm}
	\end{align}
Since the inter-target delay/Doppler separations are generally unknown a priori,
we adopt a maximum-entropy (least-informative) model and treat the relative DD offsets
$\{(\Delta\tau_{q,q'},\Delta\psi_{q,q'})\}_{q'\neq q}$ as i.i.d. samples uniformly distributed over
$\mathcal{G}\setminus\{(0,0)\}$ and independent of the data symbols.

The average sidelobe power per non-peak DD bin can be defined as
\begin{equation}\label{eq:rho_bar_MF_def}
	\vspace{-1mm}
	\bar\rho_{\mathrm{MF}}
	\triangleq
	\frac{\sum_{[k,p]\in\mathcal{G}\setminus\{(0,0)\}}
		\mathbb E\!\left[|\chi^{(0,0)}_{\mathrm{MF}}[k,p]|^2\right]}{NM-1}.
		\vspace{-1mm}
\end{equation}
Under the uniform-offset model, for any $q'\neq q$ we have the approximation
$
	\mathbb E\!\left[|\chi^{(0,0)}_{\mathrm{MF}}(\Delta\tau_{q,q'},\Delta\psi_{q,q'})|^2\right]
	\approx \bar\rho_{\mathrm{MF}},
$
and hence the monostatic self-interference term in \eqref{eq:EISL_MF_final_DD} admits the compact form
$
\sum_{\substack{q'=0, q'\ne q}}^{Q_0}|\alpha_{q'}^{(0)}|^2\,
\mathbb E\!\left[|\chi^{(0,0)}_{\mathrm{MF}}(\Delta\tau_{q,q'},\Delta\psi_{q,q'})|^2\right]
\approx
\left(\sum_{\substack{q'=0, q'\ne q}}^{Q_0}|\alpha_{q'}^{(0)}|^2\right)\bar\rho_{\mathrm{MF}}.
$
Moreover, since $\chi^{(0,0)}_{\mathrm{MF}}[k,p]$ is the unitary 2-D DFT of
$\zeta_{n,m}^{(0)}=|X_m^{(0)}[n]|^2$ with normalization $1/\sqrt{NM}$,
Parseval's theorem yields
\begin{equation}\label{eq:Parseval_MF_total}
	\begin{aligned}
		\vspace{-1mm}
		\sum_{k=0}^{N-1}\sum_{p=0}^{M-1}\mathbb E\!\left[|\chi^{(0,0)}_{\mathrm{MF}}[k,p]|^2\right]
		&=
		\sum_{n=0}^{N-1}\sum_{m=0}^{M-1}\mathbb E\!\left[|X_m^{(0)}[n]|^4\right]
		\\ &=
		M\sum_{n=0}^{N-1}\big(P_n^{(0)}\big)^2\mu^{(0)}_{4,n}.
		\vspace{-1mm}
	\end{aligned}
\end{equation}
Hence, we have
\begin{equation}\label{eq:EISL_chi_MF_closed}
	\begin{aligned}
		\vspace{-1mm}
		\bar\rho_{\mathrm{MF}}
		=
		\frac{
			M\sum_{n=0}^{N-1}\big(P_n^{(0)}\big)^2\mu^{(0)}_{4,n}
			-\mathbb E\!\left[|\chi^{(0,0)}_{\mathrm{MF}}[0,0]|^2\right]
		}{NM-1}.
		\vspace{-1mm}
	\end{aligned}
\end{equation}
where $\mathbb E[|\chi^{(0,0)}_{\mathrm{MF}}[0,0]|^2]$ is given in \eqref{eq:MF_A_00}.
Hence,
$
\bar\rho_{\mathrm{MF}}
=
\frac{
	M\sum_{n=0}^{N-1}(P_n^{(0)})^2\mu^{(0)}_{4,n}
	-\frac{M}{N}(\sum_{n=0}^{N-1}P_n^{(0)})^2
	-\frac{1}{N}\sum_{n=0}^{N-1}(P_n^{(0)})^2(\mu^{(0)}_{4,n}-1)
}{NM-1}.
$
The per-bin MF SINR in the DD domain is then
\begin{equation}
	\vspace{-1mm}
	\mathrm{SINR}^{\mathrm{MF}}_q
	= \frac{|\alpha_q^{(0)}|^2 \mathbb E \big[|\chi^{(0,0)}_{\mathrm{MF}}[0,0]|^2\big]
	}%
	{\mathrm{EISL}^{\mathrm{MF}}_q
		+ \frac{1}{N} \sigma_z^2\sum_{n=0}^{N-1}P^{(0)}_n}.
		\vspace{-1mm}
\end{equation}

\section*{Appendix B: \textsc{Proof of Proposition \ref{SensingSINRRF}}}

Using the symbol moments in \eqref{moment_math} and the independence across BSs
and TF positions, the inter-cell RF kernel power on a DD bin is
\begin{equation}
	\label{eq:RF_chi_inter_second}
	\vspace{-1mm}
	\mathbb E \big[|\chi^{(\ell,0)}_{\mathrm{RF}}[k,p]|^2\big]
	=
	\frac{1}{N}\sum_{n=0}^{N-1}
	\frac{P^{(\ell)}_n}{P^{(0)}_n}\,\mu^{(0)}_{-2,n},
	\vspace{-1mm}
\end{equation}
which is independent of $(k,p)$. Similarly, the DD-bin RF noise power on the
matched bin is
$
	\mathbb E \big[|\Lambda_{z,\mathrm{RF}}[\tau_q^{(0)},\psi_q^{(0)}]|^2\big]
	=
	\frac{\sigma_z^2}{N}\sum_{n=0}^{N-1}\frac{\mu^{(0)}_{-2,n}}{P^{(0)}_n}.
$
Under the adopted normalization and independent-symbol model, these
per-DD-bin interference and noise powers are independent of $M$, while
the desired RF power scales linearly with $M$.
At the matched DD bin, RF ideally suppresses the monostatic sidelobes for within-CP
on-grid paths. Under the normalization in \eqref{eq:Lambda_DD_def}, the desired
RF self-kernel satisfies
$
	\chi^{(0,0)}_{\mathrm{RF}}[0,0]
	=
	\frac{1}{\sqrt{NM}}\sum_{m=0}^{M-1}\sum_{n=0}^{N-1}1
	=
	\sqrt{NM},
$
and hence
$
	\mathbb E\big[|\chi^{(0,0)}_{\mathrm{RF}}[0,0]|^2\big]=NM.
$

Collecting the bistatic and coupling gains for BS-$\ell$ into
$\sum_{q'}|\alpha_{q'}^{(\ell, 0)}|^2$, the RF effective interference-plus-sidelobe
level for the $q$-th DD bin is
\begin{equation}
	\label{eq:EISL_RF_final_DD}
	\vspace{-1mm}
	\mathrm{EISL}^{\mathrm{RF}}_q
	\triangleq
	\sum_{\ell=1}^{L-1}\Big(\sum_{q'=0}^{Q_\ell}|\alpha_{q'}^{(\ell, 0)}|^2\Big)
	\mathbb E \big[|\chi^{(\ell,0)}_{\mathrm{RF}}[k,p]|^2\big].
	\vspace{-1mm}
\end{equation}
Substituting \eqref{eq:RF_chi_inter_second} into \eqref{eq:EISL_RF_final_DD},
and multiplying both numerator and denominator by $N$, yields the RF per-bin SINR in (\ref{eq:SINR_RF_final_DD}).

\section*{Appendix C: \textsc{Proof of Proposition \ref{thm:twoBS_endpoint}}}

For a fixed $N_{\mathrm{ov}}$, homogeneity across tones and Jensen's
inequality imply that the quadratic terms are minimized by equal
power within each active subset. Within the BS-symmetric
overlap-cardinality subclass, it is therefore sufficient to consider
$
P_n^{(0)}=P_n^{(1)}=x,\ \forall n\in\mathcal O;
P_n^{(\ell)}=y,\ \forall n\in\mathcal E_\ell;
P_n^{(\ell)}=0,\ \forall n\in\mathcal E_{\bar \ell}.
$
The sum-power constraint becomes
\begin{equation}
	\label{eq:xy_constraint_new}
	\vspace{-1mm}
	N_{\mathrm{ov}}x+\frac{N-N_{\mathrm{ov}}}{2}y=NP_{\mathrm{ave}}.
	\vspace{-1mm}
\end{equation}
For either BS, the denominator is
\begin{equation}
	\label{eq:D_xy_new}
	\vspace{-1mm}
	D(N_{\mathrm{ov}};x,y)
	=\Big(S_qb+\frac{\Xi}{N}\Big)N_{\mathrm{ov}}x^2
	+S_qb\frac{N-N_{\mathrm{ov}}}{2}y^2.
	\vspace{-1mm}
\end{equation}
Minimizing \eqref{eq:D_xy_new} subject to \eqref{eq:xy_constraint_new} gives
$
\frac{x}{y}=\frac{S_qb}{S_qb+\frac{\Xi}{N}},
$
and hence
\begin{equation}
	\label{eq:Dstar_Nov_new}
	\vspace{-1mm}
	D^\star(N_{\mathrm{ov}})
	=
	\frac{N^2P_{\mathrm{ave}}^2}
	{\displaystyle
		\frac{N_{\mathrm{ov}}}{S_qb+\frac{\Xi}{N}}
		+\frac{N-N_{\mathrm{ov}}}{2S_qb}}.
		\vspace{-1mm}
\end{equation}
Define
$
\Phi(N_{\mathrm{ov}})
\triangleq
\frac{N_{\mathrm{ov}}}{S_qb+\frac{\Xi}{N}}
+\frac{N-N_{\mathrm{ov}}}{2S_qb}
.
$
Then $\Phi(N_{\mathrm{ov}})$ is affine in $N_{\mathrm{ov}}$, and
$D^\star(N_{\mathrm{ov}})=N^2P_{\mathrm{ave}}^2/\Phi(N_{\mathrm{ov}})$.
Therefore:
if $\frac{\Xi}{N}>S_qb$, then $\frac{1}{S_qb+\frac{\Xi}{N}}<\frac{1}{2S_qb}$, so $\Phi(N_{\mathrm{ov}})$ decreases with $N_{\mathrm{ov}}$, and $D^\star(N_{\mathrm{ov}})$ is minimized by the smallest feasible $N_{\mathrm{ov}}$, namely $\min\mathcal F$;
if $\frac{\Xi}{N}<S_qb$, then $\Phi(N_{\mathrm{ov}})$ increases with $N_{\mathrm{ov}}$, and $D^\star(N_{\mathrm{ov}})$ is minimized at $N_{\mathrm{ov}}=N$;
This proves the theorem. 

\section*{Appendix D: \textsc{Proof of Proposition \ref{thm:RF_twoBS_endpoint}}}
\vspace{-1.5mm}

For a fixed $N_{\mathrm{ov}}$, each BS is active on
$
K=\frac{N+N_{\mathrm{ov}}}{2}
$
tones, and the masked-RF mainlobe gain is proportional to $K^2$.
Thus, up to a constant independent of $N_{\mathrm{ov}}$, the optimized
RF SINR can be written as
\begin{equation}
	\vspace{-1mm}
	\mathrm{SINR}_{\mathrm{RF}}^\star(N_{\mathrm{ov}})
	\propto
	\frac{K^2}
	{\displaystyle
		\min_{\{P_n^{(0)}\},\{P_n^{(1)}\}}\max\{D_0,D_1\}} .
	\vspace{-1mm}
\end{equation}
For the homogeneous symmetric partition, Jensen's inequality gives
\begin{equation}
	\vspace{-1mm}
	\sum\nolimits_{n\in\mathcal O\cup\mathcal E_\ell}
	\frac{1}{P_n^{(\ell)}}
	\ge
	\frac{K^2}{N P_{\mathrm{ave}}},
	\qquad \ell\in\{0,1\}.
	\vspace{-1mm}
\end{equation}
Moreover, the two ratio-type overlap terms satisfy
\begin{equation}
	\vspace{-1mm}
	\max\left\{
	\sum_{n\in\mathcal O}\frac{P_n^{(1)}}{P_n^{(0)}},
	\sum_{n\in\mathcal O}\frac{P_n^{(0)}}{P_n^{(1)}}
	\right\}
	\ge
	N_{\mathrm{ov}}.
	\vspace{-1mm}
\end{equation}
Both lower bounds are simultaneously attained by the uniform symmetric
allocation
$
P_n^{(\ell)}
=
\frac{N P_{\mathrm{ave}}}{K},
n\in\mathcal O\cup\mathcal E_\ell, \ell\in\{0,1\}.
$
Hence, we have
\begin{equation}
	\vspace{-1mm}
	\min_{\{P_n^{(0)}\},\{P_n^{(1)}\}}\max\{D_0,D_1\}
	=
	\mu_{-2}\sigma_z^2
	\frac{K^2}{N P_{\mathrm{ave}}}
	+
	\mu_{-2}\Xi N_{\mathrm{ov}}.
	\vspace{-1mm}
\end{equation}
Therefore,
$
\mathrm{SINR}_{\mathrm{RF}}^\star(N_{\mathrm{ov}})
\propto
\frac{K^2}
{\mu_{-2}\sigma_z^2\frac{K^2}{N P_{\mathrm{ave}}}
	+\mu_{-2}\Xi N_{\mathrm{ov}}}.
$
Equivalently,
$
\mathrm{SINR}_{\mathrm{RF}}^\star(N_{\mathrm{ov}})
\propto
\frac{1}
{\frac{\mu_{-2}\sigma_z^2}{N P_{\mathrm{ave}}}
	+\mu_{-2}\Xi\frac{N_{\mathrm{ov}}}{K^2}}.
$
Using $K=(N+N_{\mathrm{ov}})/2$, the overlap-dependent term becomes
$
\frac{N_{\mathrm{ov}}}{K^2}
=
\frac{4N_{\mathrm{ov}}}{(N+N_{\mathrm{ov}})^2}.
$
Therefore, for $0\le x\le N$, we have 
\begin{equation}
	\frac{d}{dx}
	\left(
	\frac{x}{(N+x)^2}
	\right)
	=
	\frac{N-x}{(N+x)^3}
	\ge 0.
\end{equation}
Thus, when $\Xi>0$, the normalized denominator in the reciprocal form of
$\mathrm{SINR}_{\mathrm{RF}}^\star(N_{\mathrm{ov}})$ is nondecreasing
in $N_{\mathrm{ov}}$, and the optimized RF SINR is strictly decreasing
for $N_{\mathrm{ov}}<N$. Hence the optimal overlap is the smallest
feasible value in $\mathcal F$. This completes the proof. 
	
\footnotesize  	
\bibliography{mybibfile}
\bibliographystyle{IEEEtran}

\end{document}